\documentclass{article}
\usepackage{arxiv}

\usepackage{mathrsfs}
\usepackage{amssymb,amsmath,amsthm}
\usepackage{mathtools}

\usepackage{tikz}
\usetikzlibrary{shapes}
\usepackage{algorithm}
\usepackage{algorithmic}

\newcommand{\FF}{\mathcal{F}}
\newcommand{\RR}{\mathscr{R}}
\newcommand{\TT}{\mathbb{T}}
\newcommand{\BB}{\mathbb{B}}

\usepackage{quoting}
\usepackage{caption}
\usepackage{enumitem}
\usepackage{geometry}
\usepackage{color}
\usepackage{multirow}
\usepackage{bigdelim}
\usepackage{subfigure}
\usepackage{hyperref}

\newtheorem{Proposition}{Proposition}
\newtheorem{Definition}{Definition}
\newtheorem{Lemma}{Lemma}
\newtheorem{Theorem}{Theorem}

\title{Computation of immediate neighbours of monotone Boolean functions}

\usepackage{authblk}

\newbox{\orcid}\sbox{\orcid}{\includegraphics[scale=0.06]{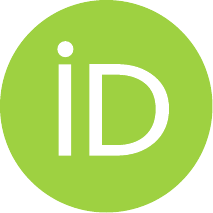}} 
\author[1]{%
	\href{https://orcid.org/0000-0001-7740-8539}{\usebox{\orcid}\hspace{1mm}
    Jos\'e E. R. Cury\thanks{These authors contributed equally to this work and share first authorship}}%
}
\author[2]{%
    Patrícia Tenera Roxo$^*$
}
\author[2]{%
	\href{https://orcid.org/0000-0002-4205-2189}{\usebox{\orcid}\hspace{1mm}
    Vasco Manquinho}%
}
\author[3]{%
	\href{https://orcid.org/0000-0003-2350-0756}{\usebox{\orcid}\hspace{1mm}Claudine Chaouiya\thanks{\texttt{Claudine.Chaouiya@univ-amu.fr}}}%
}
\author[2]{%
	\href{https://orcid.org/0000-0002-7934-5495}{\usebox{\orcid}\hspace{1mm}
    Pedro T. Monteiro\thanks{\texttt{Pedro.Tiago.Monteiro@tecnico.ulisboa.pt}}}%
}

\affil[1]{DAS, Universidade Federal de Santa Catarina, Florian\'opolis, Brazil}
\affil[2]{INESC-ID / IST - Universidade de Lisboa, Lisboa, Portugal}
\affil[3]{Aix Marseille Univ, CNRS, I2M, Marseille, France}

\begin{document}
\maketitle


\begin{abstract}
Boolean networks constitute relevant mathematical models to study the behaviours of genetic and signalling networks. These networks define regulatory influences between molecular nodes, each being associated to a Boolean variable and a regulatory (local) function specifying its dynamical behaviour depending on its regulators.
However, existing data is mostly insufficient to adequately parametrise a model, that is to uniquely define a regulatory function for each node. With the intent to support model parametrisation, this paper presents results on the set of Boolean functions compatible with a given regulatory structure, {\it i.e.} the partially ordered set of monotone non-degenerate Boolean functions. More precisely, we present original rules to obtain the direct neighbours of any function of this set.
Besides a theoretical interest, presented results will enable the development of more efficient methods for Boolean network  synthesis and revision, benefiting from the progressive exploration of the vicinity of regulatory functions.
\end{abstract}


\section{\label{sec:intro}Introduction}

{Boolean or multi-valued models} have been successful in assessing dynamical properties of {biological regulatory networks} \cite{abou-jaoude2016}. {Seminal work by S. Kauffman \cite{kauffman1969} and R. Thomas \cite{thomas1973} have been pursued with a significant range of studies leading to theoretical results or formalism extensions (e.g., among many others \cite{abou-jaoude2019,aracena2004,aracena2017,Didier:2011aa,garg2009,pauleve2020,shmulevich2002,thieffry1999,zanudo2013,remy2008}) and to computational tools for model development and analyses (e.g., \cite{boolnet,ginsim2018,colomoto2018}). There is also a long history of modelling studies for a variety of biological processes such as cell division cycle \cite{Li:2004aa,faure2007}, cell differentiation during the fly development \cite{Albert2003,sanchez2008,sanchez2001}, immune T helper cell differentiation\cite{mendoza2006,Naldi2010}, tumour cell migration \cite{Cohen2015,Selvaggio2020},  and many more. While the definition of such models does not require quantitative kinetic parameters, it still implies the specification of {the (logical) regulatory functions} to describe the combined effects of regulators upon their targets.
Data on the mechanisms underlying regulatory mechanisms are still scarce, and modellers often rely on generic regulatory functions; for instance, a component is activated if at least one activator is present and no inhibitors are present \cite{mendoza2006}, or if the weighted sum of its regulator activities is above a specific threshold (e.g., \cite{Bornholdt:2008aa,Li:2004aa}).

Here, we focus on {Boolean models}, and we address {the following} questions: 1) how complex is the parametrisation of a Boolean model consistent with the regulatory structure defined by a digraph with no multiple edges, and 2) how to modify the parametrisation so as to minimise the changes in the dynamical properties. The latter amounts to disclose the structure of the set of candidate functions that is, as argued below, the set of non-degenerate monotone Boolean functions.

Given a component $g$, element of a Boolean model, we characterise the set  $\mathcal{F}_g$ of the Boolean regulatory functions compatible with its regulatory structure, {\it i.e.,} with the number and signs of its regulators, which are either activators or inhibitors.
Generically, if $g$ has $n$ regulators, one can in principle define $2^{2^n}$ potential Boolean regulatory functions. This number is then reduced when imposing the functionality of the interactions ({\it i.e.}, all variables associated with the regulators are essential), and a fixed sign of these interactions. We focus on monotone Boolean functions \cite{thieffry1999,Gedeon2024}, {\it i.e.,} each interaction has a fixed sign (positive when its source is an activator, or negative when its source in an inhibitor), and signs of the literals correspond to those of the associated interactions.
However, there is no closed expression of the number of monotone increasing Boolean functions on $n$ variables, known as the Dedekind number \cite{korshunov2003,stephen2014}. Actually, it is even unknown for $n>9$, and its value for $n=9$ was determined only very recently \cite{vanhirtum2023computation,jakel2023}. Hence, even if the functionality constraint further restricts the number of Boolean functions compatible with a given regulatory structure, this number can still be astronomical.
The set $\mathcal{F}_g$   is the set of the monotone (positive or negative in each of its variables), non-degenerate Boolean functions. As set inclusion defines a partially ordered set  $(\mathcal{F}_g,\preceq)$ by considering the True sets of the Boolean functions, the resulting lattice can be visualised on a Hasse diagram. In this work, we propose an original algorithm to explore paths in this diagram, that is to determine the local neighbouring functions of any function in the set $(\mathcal{F}_g,\preceq)$.

Section \ref{sec:preliminaries} introduces some preliminaries.
In Section \ref{sec:po_in_f}, we characterise the set $\mathcal{F}_g$ of regulatory functions consistent with the regulatory structure of a given gene $g$.
Sections \ref{sec:parents} and \ref{sec:children} characterise the immediate neighbours of any function in $(\mathcal{F}_g,\preceq)$.
Section \ref{sec:implementation} proposes two algorithms to compute the immediate parents and children of any function in $\mathcal{F}_g$.
The paper ends with some conclusions and prospects in Section \ref{sec:conclusion}.

\section{\label{sec:preliminaries}Background}

This section includes essential definitions and notation on Boolean functions and Boolean networks. 
For further detail on the notions introduced here, we refer to relevant textbooks \cite{caspard2012}. Basics on  \emph{Partially Ordered Sets} (\emph{poset}) are provided in Appendix \ref{app:posets}. 


Considering the set $\BB = \{0,1\}$, let $\BB^{p}$ denote the set of $p$-dimensional vectors $\mathbf{x} = (x_1,\dots,x_p)$ with entries in $\BB$.

{A Boolean function} $f: \BB^{p} \rightarrow \BB$ is \emph{positive} (respectively \emph{negative}) in $x_i$ if $f|_{x_{i}=0} \le f|_{x_{i}=1}$ (respectively $f|_{x_{i}=0} \ge f|_{x_{i}=1}$), where $f|_{x_{i}=0}$ (respectively $f|_{x_{i}=1}$) denotes the value of $f(x_1,\dots,x_{i-1},0,x_{i+1},\dots,x_p)$ (respectively $f(x_1,\dots,x_{i-1},1,x_{i+1},\dots,x_p)$). 

We say that $f$ is \emph{monotone in $x_i$} if it is either positive or negative in $x_i$. It is \emph{monotone} if it is monotone in $x_i$ for all $i \in \{1,\dots,p\}$, and it is \emph{positive} (respectively \emph{negative}) if it is positive in all its variables \cite{crama_hammer}.

Determining the number $M(p)$ of positive Boolean functions for $p$ variables is known as \emph{Dedekind's problem}. This number, also called Dedekind number, is equivalent to the number of antichains in the poset $(2^{\{1,\ldots,p\}},\subseteq)$. $M(p)$ has been computed for values of $p$ up to 9, while asymptotic estimates have been proposed for higher values \cite{caspard2012}.

A variable $x_i$ is an \emph{essential} variable of a Boolean function $f$ if there is at least one $\mathbf{x} \in \BB^{p}$ such that $f|_{x_{i}=0} \ne f|_{x_{i}=1}$. A Boolean function is said to be \emph{non-degenerate} if it does not have fictitious variables, that is, all variables are \emph{essential} \cite{shmulevich2002}.

Given a Boolean function $f: \BB^{p} \rightarrow \BB$,  $\TT(f)$ denotes the set of vectors $\mathbf{x} \in \BB^{p}$ for which $f(\mathbf{x}) = 1$; in other words, $\TT(f)$ is the {\it True set} of $f$ \cite{crama_hammer,korshunov2003}.

There are many ways to represent a Boolean function. One such way is by means of a Disjunctive Normal Form (DNF). A function $f$ is said to be in DNF if it is expressed as a disjunction of conjunctions of literals. A conjunction is an \emph{implicant} of $f$ if it implies $f$. It is a \emph{prime implicant}, if it is minimal ({\it i.e.,} the removal of any literal - a non-complemented or complemented variable) results in a non-implicant of $f$. As a DNF of a function is in general not unique, we focus on the \emph{Complete Disjunctive Normal Form} (CDNF) variation, which is the disjunction of all its \emph{prime implicants}. Any Boolean function can be uniquely represented by its CDNF \cite{crama_hammer}, and in the remainder of this paper, Boolean functions are thus assumed to be expressed in their CDNFs.



 A \emph{Boolean Network} (BN) is defined by a triplet $\RR=(G, R,\mathcal{F})$, where:
\begin{itemize}
  \item $G = \{g_{i}\}_{i = 1, \ldots, n}$ is the set of $n$ regulatory components, each $g_{i}$ being associated to a Boolean variable $x_i$ in $\BB$ that denotes the activity state of $g_i$. The set $\BB^{n}$ defines the state space of $\RR$, and $\mathbf{x} = (x_1, \dots, x_{n}) \in \BB^{n}$ defines a state of the model;
  \item $R\subseteq G\times G\times\{+,-\}$ is the set of interactions, $(g_i,g_j)$, together with the effect  $g_i$ has on  $g_j$. $(g_i,g_j,+)$ denotes an activatory effect, and $(g_i,g_j,-)$ denotes an inhibitory effect of $g_j$;
  \item $\mathcal{F} = \{f_{_i}\}_{i = 1, \ldots, n}$ is the set of regulatory Boolean functions; $f_{i}: \BB^{n} \rightarrow \BB$ defines the target level of $g_{i}$ for each state $\mathbf{x} \in \BB^{n}$.
\end{itemize}

In the corresponding {\em regulatory graph} $(G,R)$, nodes represent regulatory components ({\it e.g.} genes) and directed edges represent signed regulatory interactions (positive for activations and negative for inhibitions).
Figure~\ref{fig:toyBN} shows an example of a regulatory graph with 3 components: a mutual inhibition between $g_2$ and $g_3$, and a self-activation of $g_1$, which is further activated by $g_2$ and repressed by $g_3$.

 The set of the regulators of $g_i$ is denoted $G_i\,=\,\{g_j\in G, (g_j,g_i,+)\!\in\! R \mbox{ or } (g_j,g_i,-)\!\in\! R\}$.
Note that the regulatory function of $g_i$ may be defined over the states of its regulators (rather than over the states of the full set of components): $\forall g_i\in G,\, f_i:\BB^{|G_i|}\rightarrow \BB$; it thus specifies how regulatory interactions are combined to affect the state of $g_i$. In other words, one can define the regulatory functions over only their essential variables.

\begin{figure*}[t]
  \centering
\begin{tabular}{|c|@{}c|@{}c|}\hline
\multicolumn{1}{|l|}{\bf A}&\multicolumn{1}{l|}{\bf B}&\multicolumn{1}{l|}{\bf C}\\
\begin{minipage}{0.25\textwidth}
\begin{tikzpicture}[scale=0.55]
  \path(0,0)    node[draw, fill=white!20,shape=circle](G1) {$g_1$};
  \path(4,0)    node[draw, fill=white!20,shape=circle](G2) {$g_2$};
  \path(2,-2.7) node[draw, fill=white!20,shape=circle](G3) {$g_3$};
  \draw[->, thick] (G1) edge[very thick,green, loop left] (G1);
  \draw[->, thick] (G2) edge[very thick,green,shorten >=1.5pt] (G1);
  \draw[-|, thick,shorten >=1.5pt] (G3) edge[very thick,red] (G1);
  \draw[-|, thick,shorten >=1.5pt] (G3) edge[very thick,red, bend left=30] (G2);
  \draw[-|, thick,shorten >=1.5pt] (G2) edge[very thick,red, bend left=30] (G3);
\end{tikzpicture}
\end{minipage}
&\begin{minipage}{0.29\textwidth}\centering
\resizebox{0.73\textwidth}{!}{%
\begin{tikzpicture}
  \tikzstyle{stgstate} = [draw,shape=rectangle,font=\small, fill, color=gray!20]
  \tikzstyle{stgedge} = [-latex, thick,font=\sffamily\normalsize\bfseries]

  \def\incx{1.7}
  \def\incy{1.7}
  \colorlet{darkgreen}{green!50!black}

  \node[stgstate] (000) at (0*\incx,0*\incy) {\textcolor{black}{000}};
  \node[stgstate] (100) at (1*\incx,0*\incy) {\textcolor{black}{100}};
  \node[draw,shape=rectangle,font=\small, fill, color=red!60] (001) at (0.6*\incx,0.5*\incy) {\textcolor{black}{001}};
  \node[draw,shape=rectangle,font=\small, fill, color=red!60] (101) at (1.6*\incx,0.5*\incy) {\textcolor{black}{101}};

  \node[stgstate] (010) at (0*\incx,1*\incy) {\textcolor{black}{010}};
  \node[draw,shape=rectangle,font=\small, fill, color=red!60] (110) at (1*\incx,1*\incy) {\textcolor{black}{110}};
  \node[stgstate] (011) at (0.6*\incx,1.5*\incy) {\textcolor{black}{011}};
  \node[stgstate] (111) at (1.6*\incx,1.5*\incy) {\textcolor{black}{111}};
  \draw[stgedge] (010) edge (110);
    \draw[stgedge] (000) edge (100);

  \draw[stgedge] (000) edge (010);
  \draw[stgedge] (100) edge (110);
  \draw[stgedge] (011) edge (001);
  \draw[stgedge] (111) edge (101);
  \draw[stgedge] (000) edge (001);
  \draw[stgedge] (100) edge (101);
  \draw[stgedge] (011) edge (010);
  \draw[stgedge] (111) edge (110);
\end{tikzpicture}
}
$\begin{array}{l}\\
f_1(x_1,x_2,x_3) = x_1 \vee x_2 \vee \neg  x_3\\
f_2(x_3) = \neg x_3\\
f_3(x_2) = \neg x_2
\end{array}$\end{minipage}
&\begin{minipage}{0.29\textwidth}\centering

\resizebox{0.73\textwidth}{!}{%
\begin{tikzpicture}
  \tikzstyle{stgstate} = [draw,shape=rectangle,font=\small, fill, color=gray!20]
  \tikzstyle{stgedge} = [-latex, thick,font=\sffamily\normalsize\bfseries]

  \def\incx{1.7}
  \def\incy{1.7}
  \colorlet{darkgreen}{green!50!black}

  \node[stgstate] (000) at (0*\incx,0*\incy) {\textcolor{black}{000}};
  \node[stgstate] (100) at (1*\incx,0*\incy) {\textcolor{black}{100}};
  \node[draw,shape=rectangle,font=\small, fill, color=red!60] (001) at (0.6*\incx,0.5*\incy) {\textcolor{black}{001}};
  \node[draw,shape=rectangle,font=\small, fill, color=red!60] (101) at (1.6*\incx,0.5*\incy) {\textcolor{black}{101}};

  \node[stgstate] (010) at (0*\incx,1*\incy) {\textcolor{black}{010}};
  \node[draw,shape=rectangle,font=\small, fill, color=red!60] (110) at (1*\incx,1*\incy) {\textcolor{black}{110}};
  \node[stgstate] (011) at (0.6*\incx,1.5*\incy) {\textcolor{black}{011}};
  \node[stgstate] (111) at (1.6*\incx,1.5*\incy) {\textcolor{black}{111}};
  \draw[stgedge] (010) edge (110);
  \draw[stgedge] (000) edge (010);
  \draw[stgedge] (100) edge (110);
  \draw[stgedge] (011) edge (001);
  \draw[stgedge] (111) edge (101);
  \draw[stgedge] (000) edge (001);
  \draw[stgedge] (100) edge (101);
  \draw[stgedge] (011) edge (010);
  \draw[stgedge] (111) edge (110);
\end{tikzpicture}
}
$\begin{array}{l}\\
f_1(x_1,x_2,x_3) = x_1 \vee (x_2\! \wedge \neg  x_3)\\
f_2(x_3) = \neg x_3\\
f_3(x_2) = \neg x_2
\end{array}$\vspace{0.1cm}
\end{minipage}\\\hline
\end{tabular}
\caption{Example of a Boolean Network with: (A) the regulatory graph, where normal (green) arrows represent activations and hammerhead (red) arrows represent inhibitions; (B-C) asynchronous state transition graphs, considering  Boolean regulatory functions consistent with the regulatory graph in (A). Sole the function of $g_1$ differs, leading to the loss of a transition from panel (B) to (C). Stable states (fixed points of the regulatory functions) are denoted in red.
}
\label{fig:toyBN}
\end{figure*}

A BN defines a dynamics represented by a {\em State Transition Graph} (STG), where each node represents a state $\mathbf{x}\in \BB^n$, and directed edges represent transitions between states. It depends on an updating mode, which can be synchronous, asynchronous as defined by R. Thomas \cite{thomas91a}, or others \cite{abou-jaoude2016,Le-Novere:2015aa,Pauleve2021,thomas91a}. For instance, the asynchronous STG encompasses a transition between a state $\mathbf{x}$ to a state $\mathbf{x'}$ iff 
$$\left\{\begin{array}{l}
\exists i \in \{1,\ldots,n\},\, f_i(\mathbf{x})=\neg x_i=x^{\prime}_i ,\\
  \forall j \in \{1,\ldots,n\},\, j\neq i,\,f_j(\mathbf{x})=x^{\prime}_j. \\
\end{array}\right.$$

Figure \ref{fig:toyBN} illustrates that changing the regulatory function of a component leads to the addition or the loss of transitions.

\section{\label{sec:po_in_f}Characterising the set of consistent regulatory functions}

Given a generic component $g_i$ of a BN, we first characterise the regulatory functions that comply with the interactions targeting $g_i$. We then discuss some properties of the set of such functions, as well as its cardinality.

\subsection{\label{subsec:consist_f} Consistent regulatory functions are non-degenerate monotone Boolean functions}

Consider $\RR=(G,R,\FF)$ a BN and $g_i\in G$ with  $G_i$ its set of $|G_i|$ regulators.
There are $2^{2^{|G_i|}}$ Boolean functions over the $|G_i|$ variables associated to the regulators of $g_i$. However, we can restrict ourselves to the sole functions  that comply with the regulatory structure of $g_i$, {\it i.e.,} that reflect the signs and functionalities of the regulations affecting $g_i$ \cite{abou-jaoude2016,Gedeon2024,thieffry1999}.
Recall that we consider the restricted class of BN with no dual regulations, {\it i.e.}, any regulator is either an activator  or an inhibitor.

An interaction $(g_j,g_i)$ is said to be {\em functional} if $\exists \mathbf{x} \in \mathbb{B}^n : f_i(\mathbf{x})|_{x_{j}=0}\neq f_i(\mathbf{x})|_{x_{j}=1}$, and positive (respectively, negative) if $\forall \mathbf{x} \in \mathbb{B}^n : f_i(x)|_{x_{j}=0}\leq f_i(\mathbf{x})|_{x_{j}=1}$ (respectively if $\forall \mathbf{x} \in \mathbb{B}^n : f_i(\mathbf{x})|_{x_{j}=0}\geq f_i(\mathbf{x})|_{x_{j}=1}$). Whenever $g_j\in G_i$, interaction $(g_j,g_i)$ must be functional. Moreover, as $(g_j,g_i)$ must comply with a prescribed sign, it is either positive or negative. In other words, whenever $g_j$ is a regulator of $g_i$, $x_j$ is an essential variable of $f_i$ and $f_i$ is monotone in $x_j$.

The set of regulators $G_i$ can thus be partitioned as $G_i = G^{+}_i \cup G^{-}_i$, where $G^{+}_i$ is the set of positive regulators of $g_i$ (activators), while components in $G^{-}_i$ are negative regulators of $g_i$ (inhibitors). For the example in Figure \ref{fig:toyBN}, we have: $G_{1} = \{g_1,g_2,g_3\}$, $G^{+}_{1} = \{g_1,g_2\}$ and $G^{-}_{1} = \{g_3\}$, $G_{2} = G^{-}_{2} = \{g_3\}$, $G_{3} = G^{-}_{3} = \{g_2\}$ and $G^{+}_{2} = G^{+}_{3} = \varnothing$.

Given $g_i$, the definition below characterises $\mathcal{F}^i$,  the set of all {\it consistent Boolean regulatory functions}, {\it i.e.} the functions that comply with the regulatory structure defined by ($G^{+}_i,G^{-}_i$). 

\begin{Definition}\label{def:consistency}
The set $\mathcal{F}^i$ of consistent Boolean regulatory functions of component $g_i$ is the set of non-degenerate monotone Boolean (NDMB) functions $f_i$ such that, $f_i$ is positive in $x_k$ for $g_k \in G^{+}_i$ and negative in $x_k$ for $g_k \in G^{-}_i$.
\end{Definition}

Monotonicity derives from the non-duality assumption (an interaction is either positive or negative), and the sign of the interaction from a regulator $g_k$ enforces the positiveness (if $g_k \in G^{+}_i$) or negativeness (if $g_k \in G^{-}_i$). Finally, regulatory functions must be non-degenerate due to the requirement of the functionality of all $g_k\in G_i$. We thus focus on NDMB functions describing the evolution of regulated components \cite{Gedeon2024,thieffry1999}. 

Let CDNF$(f_i)=\bigvee_{j=1}^{m} \left(\bigwedge_{k \in s_j} u_k\right)$ denote the CDNF of the regulatory function $f_i$. CDNF$(f_i)$ satisfies:

\begin{align*}
        &\text{(i) }\, \forall g_k \in G_i, \exists j \text{ such that } k \in s_j;\\
        &\text{(ii) }\, \forall j,\,\forall k\in s_j,\,u_k =
\left\{
  \begin{array}{ll}
    x_k , & \hbox{if $g_k \in G^{+}_i$,} \\
    \neg x_k , & \hbox{if $g_k \in G^{-}_i$.}
  \end{array}
\right.
\end{align*}
Both conditions $(i)$ and $(ii)$ agree with Definition \ref{def:consistency}: $(i)$ enforces the functionality of all regulators in $G_i$; $(ii)$ guarantees the consistency of the function with the sign of the regulatory interaction $(g_k,g_i)$. Note that by the definition of CDNF no two $s_j$, $s_l$ ($j \neq l$) are such that $s_j \subset s_l$.

\subsection{Set representation and number of consistent regulatory functions}
Given the regulatory structure defined by $G_i$, any NDMB function $f_i \in \mathcal{F}^i$ can be unambiguously represented by its {set-representation}, as defined below.
Necessary concepts on posets and antichains can be found in Appendix \ref{app:posets}.

\begin{Definition}
Given a component $g_i$ with $G_i = G^{+}_i\, \cup \,G^{-}_i$ its set of $p$ regulators, the set-representation $S(f_i) \subseteq 2^{\{1, \dots, |G_i|\}}$ of the regulatory function $f_i \in \mathcal{F}^i$ is such that $s_j \in S(f_i)$ if and only if $\left (\bigwedge_{k\in s_j}u_k\right )$ is a conjunctive clause of the CDNF representation of $f_i$.
\end{Definition}

Given a function $f$ in CDNF, each clause is a prime implicant of $f$, accounted for by a unique set in the set representation.
In the definition above, $S(f_i)$ represents the structure of $f_i$ as its elements indicate which variables (regulators) are involved in each of the clauses defining $f_i$. The literals are then unambiguously determined by $G^{+}_i$ and $G^{-}_i$. For example, the set-representation of $f_{1}(x_1,x_2,x_3) = x_1 \vee (x_2 \land \neg x_3)$ is $S(f_1) = \{\{1\},\{2,3\}\}$.

Since elements of $S(f_i)$ are pairwise incomparable subsets of $\{1,\dots, |G_i|\}$, for the $\subseteq$ relation, it is easy to verify that $S(f_i)$ is an antichain in the {poset} $(2^{\{1,\dots, |G_i|\}},\subseteq)$. Moreover, $S(f_i)$ is also a cover of $\{1,\dots, |G_i|\}$ since all indices must appear in at least one element of $S(f_i)$. Finally, any antichain in $(2^{\{1,\dots, |G_i|\}},\subseteq)$ which is a cover of $\{1,\dots, |G_i|\}$ {is the set representation} of a unique function in $\mathcal{F}^i$. Therefore, $\mathcal{F}^i$, the set of consistent Boolean regulatory functions of $g_i$, is isomorphic to $\mathcal{S}_{|G_i|}$, the set of antichains in $(2^{\{1,\dots, |G_i|\}},\subseteq)$ that are also a cover of $\{1,\dots, |G_i|\}$.

As "when a monotone function is neither positive nor negative, it can always be brought to one of these two forms by an elementary change of variables" \cite{crama_hammer}, from now on, we will restrict ourselves to (non-degenerated monotone) positive functions. We will also denote $\mathcal{F}_p$ the set of such functions over $p$ variables. 

\begin{figure*}[t]
\resizebox{\textwidth}{!}{%
\begin{tabular}{c|r|r}
    $p$ & \multicolumn{1}{c}{$M(p)$} &\multicolumn{1}{c}{$N(p)=|\mathcal{F}_p| = |\mathcal{S}_p|$}\\ \hline
    1 & 3&1\\
    2 & 6 & 2\\
    3 & 20 & 9\\
    4 & 168 & 114\\
    5 & 7 581 & 6 894\\
    6 & 7 828 354 & 7 785 062\\
    7 & 2 414 682 040 998 & 2 414 627 396 434\\
    8 & 56 130 437 228 687 557 907 788 & 56 130 437 209 370 320 359 968\\
    9 & 286 386 577 668 298 411 128 469 151 667 598 498 812 366 &
        286 386 577 668 298 410 623 295 216 696 338 374 471 993
  \end{tabular}}
  \caption{Numbers of positive Boolean functions (Dedekind number $M(p)$) and of non-degenerate  positive Boolean functions ($N(p)$) of $p=1,\dots, 9$ variables. $N(p)$ is also the number of antichain covers of $\{1,\dots, p\}$.\label{fig:nfunctions}}
\end{figure*}

The cardinality $N(p)$ of $\mathcal{F}_p$, set of all \textit{non-degenerate monotone positive Boolean} (NDPB) functions of $p$ variables, is smaller than $2^{2^{p}}$, the number of all Boolean functions of $p$ variables and also than $M(p)$, the Dedekind number of monotone positive Boolean functions (including degenerate functions.
Indeed, one can easily show that:
 $$N(p)=M(p)-2-\sum_{k=1}^{p-1}\frac{p!}{k!(p-k)!}N(k).$$

Nevertheless, as illustrated in Figure \ref{fig:nfunctions}, the cardinality of $\mathcal{F}_p$ dramatically increases with the number of variables ($p$ regulators) and thus constitutes a major computational challenge.
Any approach relying on the exploration of the full set $\mathcal{F}_p$ where $p > 5$ would be intractable.
In this context, the possibility to iteratively navigate within $\mathcal{F}_p$ is crucial to assess the impact of particular regulatory functions on the dynamics of the corresponding BN, in a computational tractable manner.

\subsection{Partially Ordered Set of non-degenerate monotone Boolean functions}
Considering $\mathcal{F}_p$ the set of positive NDMB functions over $p$ variables, we show that it is clearly a {poset} by considering the binary relation $\preceq$ on $\mathcal{F}_p \times \mathcal{F}_p$ defined by:

$$\forall f,f^\prime \in \mathcal{F}_p,\, f \preceq\ f^{\prime} \iff \TT(f) \subseteq \TT(f^{\prime}).$$

It is easy to verify that $(\mathcal{F}_p,\preceq)$ is a {poset}. Figure \ref{fig:HD_3} shows the Hasse Diagram (HD) of the {poset} $(\mathcal{F}_{|G_1|},\preceq)$ of $g_1$, a component of the model presented in Figure \ref{fig:toyBN}.

\begin{figure}[!ppt]
  \centering
  \resizebox{0.5\textwidth}{!}{%
\begin{tikzpicture}
  \tikzstyle{hdstate} = [ellipse,draw,blue,font=\bfseries\small]
  \tikzstyle{hdedge} = [gray,font=\sffamily\normalsize\bfseries]
  \tikzstyle{hdfunc} = [red,font=\large]

  \def\xspace{4}
  \def\yspace{1.5}

  \node[hdstate] (1-2-3) at (0,4*\yspace) {\{\{1\},\{2\},\{3\}\}};

  \node[hdstate] (3-12) at (-1*\xspace,3*\yspace) {\{\{3\},\{1,2\}\}};
  \node[hdstate] (2-13) at (0,3*\yspace) {\{\{2\},\{1,3\}\}};
  \node[hdstate] (1-23) at (1*\xspace,3*\yspace) {\{\{1\},\{2,3\}\}};

  \node[hdstate] (12-13-23) at (0,2*\yspace) {\{\{1,2\},\{1,3\},\{2,3\}\}};

  \node[hdstate] (12-23) at (-1*\xspace,1*\yspace) {\{\{1,2\},\{2,3\}\}};
  \node[hdstate] (12-13) at (0,1*\yspace) {\{\{1,2\},\{1,3\}\}};
  \node[hdstate] (23-13) at (1*\xspace,1*\yspace) {\{\{1,3\},\{2,3\}\}};

  \node[hdstate] (123) at (0,0) {\{\{1,2,3\}\}};

  \draw[hdedge] (3-12) edge node[black] {} (1-2-3);
  \draw[hdedge] (2-13) edge node[black] {} (1-2-3);
  \draw[hdedge] (1-23) edge node[black] {} (1-2-3);

  \draw[hdedge] (12-13-23) edge node[black] {} (3-12);
  \draw[hdedge] (12-13-23) edge node[black] {} (2-13);
  \draw[hdedge] (12-13-23) edge node[black] {} (1-23);

  \draw[hdedge] (12-23) edge node[black] {} (12-13-23);
  \draw[hdedge] (12-13) edge node[black] {} (12-13-23);
  \draw[hdedge] (23-13) edge node[black] {} (12-13-23);

  \draw[hdedge] (123) edge node[black] {} (12-23);
  \draw[hdedge] (123) edge node[black] {} (12-13);
  \draw[hdedge] (123) edge node[black] {} (23-13);

  \path[hdfunc] (1-2-3) ++(3.8,0) node {sup $F_{g_1} = s_1 \vee s_2 \vee \neg s_3$};
  \path[hdfunc] (123) ++(3.2,0) node {inf $F_{g_1} = s_1 \wedge s_2 \wedge \neg s_3$};
\end{tikzpicture}
  }%
  \caption{Hasse Diagram representing the set of all possible functions composed of 3 regulators ({\it e.g.} functions in red of the component $g_1$ of the model in Figure \ref{fig:toyBN}).}\label{fig:HD_3}
\end{figure}

Observe that, while the functions in $\mathcal{F}_p$ depend on the specific regulatory structure ({\it i.e.}, the signs of the regulations), the topology of the {HD} and the relation between its nodes, when seen as set-representations, only depend on $p$, the number of regulators of $g_i$.  In other words, the HD shown in Figure \ref{fig:HD_3} represents the set of consistent regulatory functions for any component with 3 regulators.

In fact, one can consider the relation $\preceq$ on the set $\mathcal{S}_p$ of antichains in $(2^{\{1,\dots, p\}},\subseteq)$:
$$\forall S, S^\prime \in \mathcal{S}_p,\, S\preceq S^{\prime} \iff \forall s \in S,\, \exists s^{\prime} \in S^{\prime} \text{ such that } s^{\prime} \subseteq s .$$

The set $s^\prime$ is said to be a \emph{witness} in $S^\prime$ for $s$. The above equivalence can be restated as follows: $S \preceq S^\prime$, if and only if, every set in $S$ has a witness in $S^\prime$.

Recall that $(\mathcal{S}_p,\preceq)$ is also a {poset}. Its {HD} has the same structure as the HD of $(\mathcal{F}_p,\preceq)$, where its nodes are the set-representations $S(f) \in \mathcal{S}$. This is because:
\begin{equation} \label{eq:two_relations}
f \preceq f^{\prime} \iff S(f) \preceq S(f^{\prime}).
\end{equation}

Summarising, the {poset} $(\mathcal{S}_p,\preceq)$ can be used as a template for all {posets} $(\mathcal{F}_p,\preceq)$ of regulatory functions of a component with $p$ regulators, considering any possible regulatory structures, {\it i.e.}, all pairs $(G^{+}_i,G^{-}_i)\in G$. In what follows, properties of {posets} $(\mathcal{F}_p,\preceq)$ will thus be derived from those of $(\mathcal{S}_p,\preceq)$.

Given a generic component with $p$ regulators, we introduce the following terminology on the relationships between elements in the HD of the {poset} $(\mathcal{S}_p,\preceq)$. This terminology also applies to $(\mathcal{F}_p,\preceq)$). 

\begin{def}\label{def:parents_etc}
Given $S, S^\prime \in \mathcal{S}_p$:
\begin{itemize}
\item $S^\prime$ is a \emph{parent} of $S$ in $(\mathcal{S}_p,\preceq)$ if and only if $S \prec S^\prime$;
\item $S^\prime$ is an \emph{immediate parent} of $S$ in $(\mathcal{S}_p,\preceq)$ if and only if $S^\prime$ is a parent of $S$ and $\nexists S^{\prime\prime} \in \mathcal{S}_p$ such that $S \!\prec\! S^{\prime\prime}\! \!\prec\! S^\prime$;
\item $S^\prime$ is a \emph{child} of $S$ in $(\mathcal{S}_p,\preceq)$ if and only if $S^\prime \prec S$;
\item $S^\prime$ is an \emph{immediate child} of $S$ in $(\mathcal{S}_p,\preceq)$ if and only if $S^\prime$ is a {child} of $S$ and $\nexists S^{\prime\prime} \in \mathcal{S}_p$ such that $S^\prime \prec S^{\prime\prime} \prec S$.
\end{itemize}
\end{def}

The next two sections provide rules to determine the sets of the immediate neighbours of elements in $(\mathcal{S}_p,\preceq)$. 

\section{Immediate parents of an element of $(\mathcal{S}_p,\preccurlyeq)$\label{sec:parents}}
\subsection{Rules to compute immediate parents}


Given an element $S$ of $\mathcal{S}_p$, a parent $S^\prime$ of $S$ is obtained by applying one of the following rules (example in Figure~\ref{fig:rules-example} left).

\noindent\fbox{%
	\begin{minipage}{\textwidth}
	\begin{center}{\sc rules to compute parents}\end{center}
	\begin{enumerate}[start=1,label={\hspace*{0.0cm} Rule \arabic*:},wide =0pt,leftmargin=0.2cm,parsep=3pt,topsep=0pt,font=\sc]
	
			\item[Rule 1] $S^\prime = S \cup \{\sigma\}$ where $\sigma \subset \{1,\dots p\}$ is a maximal set independent of $S$.
      \item[Rule 2] $S^\prime =(S\setminus \{s_1, \dots, s_k\}) \cup \{\sigma\}$, where:
      \begin{enumerate}[parsep=3pt,topsep=0pt]
			\item $\sigma$ is a maximal set dominated by $S$;
      \item $\sigma$ is not contained in any maximal set independent of $S$;
			\item $\forall i = 1, \ldots, k$, $\sigma$ is contained in $s_i$;
			\item $S^\prime$ is a cover of $\{1,\dots,p\}$. 
			
		\end{enumerate}

		\item[Rule 3] $S^\prime =(S\setminus \{s\}) \cup \{\sigma_1, \sigma_{2}\}$, where:
	
		\begin{enumerate}[parsep=3pt,topsep=0pt]
      \item Both $\sigma_1$ and $\sigma_2$ are maximal sets dominated by $S$;
      \item Both $\sigma_1$ and $\sigma_2$ are not contained in any maximal set independent of $S$;
      \item Both $\sigma_1$ and $\sigma_2$ are contained in $s$;
      \item Neither $(S\setminus \{s\}) \cup \{\sigma_1\}$ nor $(S\setminus \{s\}) \cup \{\sigma_2\}$ are a cover of $\{1,\dots,p\}$.
		\end{enumerate}
	\end{enumerate}
\end{minipage}}
\smallskip

\begin{Theorem}\label{theorem:parents}
$S'$ is an immediate parent of $S$ in $\mathcal{S}_p$ if and only if $S'$ is generated by one of the 3 rules to compute parents presented above.
\end{Theorem}

Proof of Theorem \ref{theorem:parents} can be found in Appendix \ref{app:proof_th1}. It first ensures that a set resulting from any of the three rules is indeed a parent of the starting set $S$. It remains to prove that it is an immediate parent of $S$. For Rule 1, immediacy is easily proved by contradiction. For Rules 2 and 3, we first prove by contradiction that for any set $S'$ generated by these rules, if a set $S''$ is such that $S \preccurlyeq S'' \preccurlyeq S'$ then $S\not\subseteq S^{\prime\prime}$. This is then used to prove, again by contradiction, that $S'$ is an immediate parent of $S$. Finally, it remains to prove that any immediate parent is generated by one of the three rules. This is proved by considering two cases: an immediate parent not generated by our rules contains $S$ or not. In the first case, a contradiction easily arises. In the second case, the proof considers a set $s$ element of $S$ that is not in $S'$ and its witness $s'$ in $S'$ ({\emph i.e.} $s'\subset s$). Then, reasoning on the possible forms of $S$, a contradiction arises.

\begin{figure*}[t]
	\begin{center}
		\begin{tabular}{|c|c|}\hline
\begin{minipage}{0.46\textwidth}\centering {\small 
$\begin{array}{l@{\hspace{0cm}}r@{\hspace{0cm}}ll}
{\rm set}&\multicolumn{3}{l}{ S_1 = \{\{1,2,3\},\{3,4\}\}} \\
\multirow{2}{*}{parents}&\ldelim \{ {2}{3 mm} 
 & \{\{1,2,3\},\{1,2,4\},\{3,4\}\}&Rule1\\
& & \{\{1,3\},\{2,3\},\{3,4\}\}&Rule3\\
\end{array}$}
\end{minipage}
&
\begin{minipage}{0.45\textwidth}\resizebox{\textwidth}{!}{
$\begin{array}{l@{\hspace{0.1cm}}r@{\hspace{0.1cm}}ll}
{\rm set}& \multicolumn{3}{l}{S_2 = \{\{1,2,3\},\{1,2,4\},\{3,4\}\}} \\
\multirow{3}{*}{children}&\ldelim \{ {3}{4 mm} 
 &\{\{1,2,4\},\{3,4\}\}&Rule1\\
& &\{\{1,2,3\},\{3,4\}\}&Rule1\\
& &\{\{1,2,3\},\{1,2,4\},\{1,3,4\},\{2,3,4\}\}&Rule2\\
\end{array}$}
\end{minipage}\\
\begin{minipage}{0.45\textwidth}
	\resizebox{\textwidth}{!}{ 
		\begin{tikzpicture}
			\tikzstyle{term} = [draw,shape=rectangle,minimum height=2em,text centered, fill, color=gray!20]
			\tikzstyle{term2} = [draw,shape=rectangle,minimum height=2em,text centered, fill, color=orange]
			\tikzstyle{term32} = [draw,ellipse,minimum height=2em,text centered, fill, color=green!50]
			\tikzstyle{term4} = [draw,shape=rectangle,minimum height=2em,text centered, fill, color=violet!50]
			\tikzstyle{term42} = [draw,ellipse,minimum height=2em,text centered, fill, color=violet!50]
			\tikzstyle{term5} = [draw,shape=rectangle,minimum height=2em,text centered, fill, color=blue!50]
			\tikzstyle{link} = [-]
			\def\xspace{1.5}
			\def\yspace{-1.5}
			
			\node[term] (1234) at (0,0){\textcolor{black}{$\{1,2,3,4\}$}};
			\node[term2](123) at (-2.5*\xspace,\yspace) {\textcolor{black}{$\{1,2,3\}$}};
			\node[term32](124) at (-1*\xspace,\yspace) {\textcolor{black}{$\{1,2,4\}$}};
			\node[term2](134) at (1*\xspace,\yspace) {\textcolor{black}{$\{1,3,4\}$}};
			\node[term2](234) at (2.5*\xspace,\yspace) {\textcolor{black}{$\{2,3,4\}$}};
			\node[term4](12) at (-4*\xspace,2*\yspace) {\textcolor{black}{$\{1,2\}$}};	
			\node[term42](13) at (-2.5*\xspace,2*\yspace) {\textcolor{black}{$\{1,3\}$}};
			\node[term4](14) at (-1*\xspace,2*\yspace) {\textcolor{black}{$\{1,4\}$}};
			\node[term42](23) at (1*\xspace,2*\yspace) {\textcolor{black}{$\{2,3\}$}};
			\node[term4](24) at (2.5*\xspace,2*\yspace) {\textcolor{black}{$\{2,4\}$}};	
			\node[term42](34) at (4*\xspace,2*\yspace) {\textcolor{black}{$\{3,4\}$}};	
			\node[term4](1) at (-2.5*\xspace,3*\yspace) {\textcolor{black}{$\{1\}$}};
			\node[term4](2) at (-1*\xspace,3*\yspace) {\textcolor{black}{$\{2\}$}};
			\node[term4](3) at (1*\xspace,3*\yspace) {\textcolor{black}{$\{3\}$}};
			\node[term4](4) at (2.5*\xspace,3*\yspace) {\textcolor{black}{$\{4\}$}};
			\node[term4](0) at (0,4*\yspace) {\textcolor{black}{$\{\emptyset\}$}};
			
			\draw[link] (1234) edge (123);
			\draw[link] (1234) edge (124);
			\draw[link] (1234) edge (134);	
			\draw[link] (1234) edge (234);	
			\draw[link] (123) edge (12);
			\draw[link] (123) edge (13);
			\draw[link] (123) edge (23);
			\draw[link] (124) edge (12);
			\draw[link] (124) edge (14);
			\draw[link] (124) edge (24);	
			\draw[link] (134) edge (13);
			\draw[link] (134) edge (14);
			\draw[link] (134) edge (34);
			\draw[link] (234) edge (23);
			\draw[link] (234) edge (34);
			\draw[link] (234) edge (24);	
			\draw[link] (1) edge (12);
			\draw[link] (1) edge (13);
			\draw[link] (1) edge (14);
			\draw[link] (2) edge (12);
			\draw[link] (2) edge (23);
			\draw[link] (2) edge (24);
			\draw[link] (3) edge (13);
			\draw[link] (3) edge (23);
			\draw[link] (3) edge (34);
			\draw[link] (4) edge (14);
			\draw[link] (4) edge (24);
			\draw[link] (4) edge (34);
			\draw[link] (0) edge (1);
			\draw[link] (0) edge (2);
			\draw[link] (0) edge (3);
			\draw[link] (0) edge (4);
	\end{tikzpicture}}
\end{minipage}& \begin{minipage}{0.45\textwidth}
	\resizebox{\textwidth}{!}{ 
		\begin{tikzpicture}
			\tikzstyle{term} = [draw,shape=rectangle,minimum height=2em,text centered, fill, color=gray!20]
			\tikzstyle{term2} = [draw,shape=rectangle,minimum height=2em,text centered, fill, color=orange]
			\tikzstyle{term3} = [draw,shape=rectangle,minimum height=2em,text centered, fill, color=green!50]
			\tikzstyle{term32} = [draw,ellipse,minimum height=2em,text centered, fill, color=green!50]
			\tikzstyle{term4} = [draw,shape=rectangle,minimum height=2em,text centered, fill, color=violet!50]
			\tikzstyle{term42} = [draw,ellipse,minimum height=2em,text centered, fill, color=violet!50]
			\tikzstyle{term5} = [draw,shape=rectangle,minimum height=2em,text centered, fill, color=blue!50]
			\tikzstyle{link} = [-]
			\def\xspace{1.5}
			\def\yspace{-1.5}
			
			\node[term] (1234) at (0,0){\textcolor{black}{$\{1,2,3,4\}$}};
			\node[term2](123) at (-2.5*\xspace,\yspace) {\textcolor{black}{$\{1,2,3\}$}};
			\node[term32](124) at (-1*\xspace,\yspace) {\textcolor{black}{$\{1,2,4\}$}};
			\node[term](134) at (1*\xspace,\yspace) {\textcolor{black}{$\{1,3,4\}$}};
			\node[term](234) at (2.5*\xspace,\yspace) {\textcolor{black}{$\{2,3,4\}$}};
			\node[term4](12) at (-4*\xspace,2*\yspace) {\textcolor{black}{$\{1,2\}$}};	
			\node[term42](13) at (-2.5*\xspace,2*\yspace) {\textcolor{black}{$\{1,3\}$}};
			\node[term3](14) at (-1*\xspace,2*\yspace) {\textcolor{black}{$\{1,4\}$}};
			\node[term42](23) at (1*\xspace,2*\yspace) {\textcolor{black}{$\{2,3\}$}};
			\node[term3](24) at (2.5*\xspace,2*\yspace) {\textcolor{black}{$\{2,4\}$}};	
			\node[term2](34) at (4*\xspace,2*\yspace) {\textcolor{black}{$\{3,4\}$}};	
			\node[term4](1) at (-2.5*\xspace,3*\yspace) {\textcolor{black}{$\{1\}$}};
			\node[term4](2) at (-1*\xspace,3*\yspace) {\textcolor{black}{$\{2\}$}};
			\node[term4](3) at (1*\xspace,3*\yspace) {\textcolor{black}{$\{3\}$}};
			\node[term4](4) at (2.5*\xspace,3*\yspace) {\textcolor{black}{$\{4\}$}};
			\node[term4](0) at (0,4*\yspace) {\textcolor{black}{$\{\emptyset\}$}};
			
			\draw[link] (1234) edge (123);
			\draw[link] (1234) edge (124);
			\draw[link] (1234) edge (134);	
			\draw[link] (1234) edge (234);	
			\draw[link] (123) edge (12);
			\draw[link] (123) edge (13);
			\draw[link] (123) edge (23);
			\draw[link] (124) edge (12);
			\draw[link] (124) edge (14);
			\draw[link] (124) edge (24);	
			\draw[link] (134) edge (13);
			\draw[link] (134) edge (14);
			\draw[link] (134) edge (34);
			\draw[link] (234) edge (23);
			\draw[link] (234) edge (34);
			\draw[link] (234) edge (24);	
			\draw[link] (1) edge (12);
			\draw[link] (1) edge (13);
			\draw[link] (1) edge (14);
			\draw[link] (2) edge (12);
			\draw[link] (2) edge (23);
			\draw[link] (2) edge (24);
			\draw[link] (3) edge (13);
			\draw[link] (3) edge (23);
			\draw[link] (3) edge (34);
			\draw[link] (4) edge (14);
			\draw[link] (4) edge (24);
			\draw[link] (4) edge (34);
			\draw[link] (0) edge (1);
			\draw[link] (0) edge (2);
			\draw[link] (0) edge (3);
			\draw[link] (0) edge (4);
	\end{tikzpicture}}
\end{minipage}
\\\hline
\end{tabular}
	\caption{Illustration of the rules to compute parents (left) and children (right), using the set representation on the Hasse diagram of the poset $(2^{\{1,2,3,4\}},\subseteq)$.
 Orange vertices depict the elements of the starting sets ($S_1$ and $S_2$); green vertices depict independent sets, among which oval ones are maximal; violet vertices depict dominated sets, among which oval ones are maximal and not included in any independent set.
 \textbf{Parents example:} $S_1$ is composed of two sets, in orange, which dominate the violet sets and are dominated by the gray sets, leaving a single $\sigma=\{1,2,4\}$ as a maximal set independent of $S_1$. Thus, by Rule 1, $S_1 \cup \{\sigma\}$ generates a new parent function.
 Additionally, $\sigma_1 = \{1,3\}$ and $\sigma_2 = \{2,3\}$ are maximal sets dominated by $S_1$, are not contained in any maximal set independent of $S_1$, are both contained in $s = \{1,2,3\}$, and neither $\{3,4\} \cup \sigma_1$ nor $\{3,4\} \cup \sigma_2$ are a cover, thus satisfying all conditions to generate a new parent function $(S_1 \setminus \{s\}) \cup \{\sigma_1, \sigma_2\}$, by Rule 3.
 \textbf{Children example:} By Rule 1, $S_2 \setminus \{s\}$, with $s = \{1,2,3\}$, is a valid child of $S_2$ since it is a cover and $s$ is a maximal set independent of $S_2 \setminus \{s\}$. The same is true for $s = \{1,2,4\}$. Finally, by taking out $s=\{3,4\}$, by Rule 2, we must add both $\sigma_1=\{1,3,4\}$ and $\sigma_2=\{2,3,4\}$, which extend $s$ by one element yielding to an independent set of $S_2 \setminus \{s\}$.  \label{fig:rules-example}}
	\end{center}
\end{figure*}


\subsection{Cardinality of the True set of an immediate parent}
\label{sec:properties-parent}

Recall that, by definition, $f \preccurlyeq f^\prime$ if and only if $\TT(f) \subseteq \TT(f^\prime)$. This section quantifies $|\TT(f^\prime) \setminus \TT(f)|$, where $f^\prime$ is an immediate parent of $f$. Each set appearing in $S(f)$ (set representation of $f$), corresponds to a prime implicant of that function. Therefore, by assigning the literals appearing in that set to True, we can enquire about the cardinality of the True set of an element of $\FF_p$ directly from its set representation. More precisely, $\mathbf{x} \in \TT(f)$ if and only if $\exists s \in S(f): \forall i \in s, x_i=1$ (recall that we consider positive functions). In other words, $\mathbf{x}$ belongs to the True set of $f$ if and only if there is at least one prime implicant in the set representation of $f$ that testifies this fact.

When it is clear from context, we may use the set representation of a function or clause in place of the function or clause itself.
\medskip

\begin{Proposition}\label{prop:rule1-parent}If $S'$ is an immediate parent of $S$ generated by Rule 1, then $|\mathbb{T}(S'\setminus S)| = 1$.
\end{Proposition}
\begin{proof}
Let $S'=S\cup \{\sigma\}$ be an immediate parent of $S$ derived from Rule 1 ($\sigma$ is a maximal set independent of $S$). Because $S\subseteq S'$, and because $\sigma$ is independent of $S$, we have that $|\TT(S)| < |\TT(S')|$. Furthermore, as $\sigma$ is maximal, we have that $\forall k\not\in \sigma,\,\exists s\in S$ such that $s\subseteq \sigma\cup\{k\}$.
Therefore $\TT(\sigma\cup\{k\})\subseteq \TT(\{s\})$, which means that every state $\mathbf{x}$ satisfying $\sigma\cup \{k\}$ (\textit{i.e.}, evaluating $\sigma\cup \{k\}$ to True) also satisfies $s$.
Hence, the only $\mathbf{x}\in\TT(S')$ such that $\mathbf{x}\not\in \TT(S)$ verifies $\forall i\in\sigma, x_i=1, \forall i\not\in \sigma, x_i=0$. This proves that $|\TT(S'\setminus S)|=1$.
\end{proof}

\begin{Proposition}\label{prop:rule2-parent}If $S'$ is an immediate parent of $S$ generated by Rule 2, then $|\mathbb{T}(S'\setminus S)| = 1$.
\end{Proposition}
\begin{proof}
Let $S'=S\setminus\{s_1,\ldots,s_k\}\cup \{\sigma\}$ be an immediate parent of $S$ derived from Rule 2, where $\sigma$ is a maximal set dominated by $S$, and not contained in any maximal set independent of $S$.
For all $i\in\{1,\ldots, k\}$, we have that $|\sigma|=|s_i|-1$, because $\sigma$ is a maximal set dominated by $\{s_1,\ldots, s_k\}$.
Without loss of generality, let $\sigma=\{1,\ldots, q\}$, and consider $\sigma'=\sigma\cup\{j\}$ with $j\in\{q+1,\ldots, p\}$.
If $\exists i\in \{1,\ldots, k\}$ such that  $\sigma' = s_i$, then $\TT(\sigma')\subset \TT(S)$.
Otherwise, if $\exists j\in\{q+1,\ldots, p\}$ such that $\sigma'= (\sigma\cup \{j\}) \not\in \{s_1,\ldots, s_k\}$ then: either (a) $\exists s\in S\setminus \{s_1,\ldots, s_k\}$ such that $s \subsetneq \sigma'$ and thus $\TT(\sigma')\subset \TT(S)$ or, (b) $\forall s\in S, \sigma'\not\subset s$, which contradicts the fact that $\sigma$ is a maximal set independent of $S$. Hence, the True set of any set defined as $\sigma$ augmented by one element of $\{q+1,\ldots, p\}$ is included in $\TT(S)$, which means that $\sigma=\{1,\ldots, q\}$  adds a single True state to $\TT(S)$, which is $\mathbf{x}\in \BB^p, \forall i\in \{1,\ldots, q\}, x_i=1$, and $\forall i\in\{q+1,\ldots, p\}, x_i=0$, in other words $|\TT(S'\setminus S)|=1$.
\end{proof}

\begin{Proposition}\label{prop:rule3-parent}If $S'$ is an immediate parent of $S$ generated by Rule 3, then $|\mathbb{T}(S'\setminus S)| = 2$.
\end{Proposition}

\begin{proof}
Let $S'=S\setminus\{s\}\cup \{\sigma_1,\sigma_2\}$ be an immediate parent of $S$ derived from Rule 3, where $\sigma_1$ and $\sigma_2$ are distinct maximal sets dominated by $S$, not contained in any maximal set independent of $S$, and contained in $s$.
Without loss of generality, let $\sigma_1=\{1,\ldots, q,q+1\}$, $\sigma_2=\{1,\ldots, q,q+2\}$, and $s=\sigma_1\cup \sigma_2=\{1,\ldots, q,q+1,q+2\}$.
Similar to the previous proofs, we will show that a single state is added by $\sigma_1$ (respectively by $\sigma_2$) to the $\mathbb{T}(S)$.

Let us consider $\sigma' =\sigma_1\cup\{i\}$ with $i\in\{q+2,\ldots, p\}$. There must exist a set $s'\in S$ such that  $s'\subseteq \sigma' $. Otherwise, there would be a maximal set independent of $S$ containing $\sigma_1$, a contradiction with Rule 3 item (b). Hence, $\TT(\sigma')\subseteq \TT(S)$, 
that is the True set of any set defined as  $\sigma_1$ augmented by one element of $\{q+2,\ldots,p\}$ is included in $\TT(s)$.

A similar argument can be made for $\sigma_2$. Since both $\sigma_1$ and $\sigma_2$  differ from $s$ by a single distinct element, this proves that $|\TT(S'\setminus S)|+2$, where the two added states in $\TT(S')$ are: (1) $\mathbf{x}\in\BB^p$ such that $ \forall i\in\{1,\ldots, q+1\}, x_i=1$ and $\forall i\in\{q+2,\ldots, p\}, x_i=0$, and (2) $\mathbf{x'}\in\BB^p$ such that $\forall i\in\{1,\ldots, q, q+2\}, x'_i=1$  and $\forall i\in\{q+1, q+3,\ldots,p\}, x'_i=~0$.
\end{proof}

\section{Immediate children of an element of $(\mathcal{S}_p,\preccurlyeq)$\label{sec:children}}

\subsection{Rules to compute immediate children}
\label{sec:rules-child}

In turn, given an element $S$ of the {poset} $(\mathcal{S}_p,\preceq)$, a child $S^\prime$ of $S$ is obtained by applying one of the following rules (illustrated in Figure~\ref{fig:rules-example} right).
Figure \ref{fig:HD-rules} Appendix \ref{sec:app:examples} provides a further illustration of the rules to compute parents and children.

\noindent\fbox{%
	\begin{minipage}{\textwidth}
		\begin{center}{\sc Rules to compute children}\end{center}	
		
		\begin{enumerate}[start=1,label={\hspace*{0.0cm}\sc Rule \arabic*:},wide =0pt,leftmargin=0.2cm, parsep=3pt,topsep=0pt]
			\item $S^\prime = S \setminus \{s\}$ where:
      \begin{enumerate}
        \item $S \setminus \{s\}$ yields a cover of $\{1,\dots, p\}$;
        \item $s$ is a maximal set independent of $S\setminus\{s\}$.
      \end{enumerate}
			
			\item $S^\prime = (S \setminus \{s\}) \bigcup_{i=1\ldots k} \{ \sigma_i \}$ with $k \geq 1$ where:
			\begin{enumerate}[parsep=3pt,topsep=0pt]
        \item $\forall i=1,\ldots, k, \sigma_i = (s \cup \{l_i\})$ with $l_i \in \{1,\ldots, p\}$ and is independent of $S \setminus \{s\}$;
        \item All $\sigma_i$ complying with $(a)$ are present in $S'$.
			\end{enumerate}

			\item $S' = (S \setminus \{s_i, s_j\}) \cup \{ \sigma \}$ where:
			\begin{enumerate}[parsep=3pt,topsep=0pt]
        \item $\sigma=s_i \cup s_j$;
        \item $s_i$ is a maximal set independent of $S\setminus\{s_i\}$;
        \item $s_j$ is a maximal set independent of $S\setminus\{s_j\}$;
        \item $|s_i \setminus s_j| = 1 = |s_j \setminus s_i| $;
  			\item Neither $S\setminus\{s_i\}$ nor $S\setminus\{s_j\}$ yields a cover of $\{1,\ldots, p\}$.
			\end{enumerate}
		\end{enumerate}
\end{minipage}}

\begin{Lemma}\label{lemma:duality}
The set $S^\prime$ defined from $S\in\mathcal{S}_p$ by Rule X to compute children above (X=1,2,3) is such that $S$ is an immediate parent of $S^\prime$ as specified by the corresponding Rule X to compute parents.
\end{Lemma}

\begin{proof}
It is clear that each proposed form of children leads to a set $S^\prime$ that is an element of $\mathcal{S}_p$ because it originates from $S$ by removing one or two elements and possibly adding others that are non-comparable with the remaining elements. Furthermore, by checking the removed and added sets involved in constructing $S^\prime$ it is easy to see that $S^\prime$ is a cover of $\{1,\ldots,p\}$ and that $S' \preccurlyeq S$.
  
Let $S^\prime = S \setminus \{s\}$ be a set defined by Rule 1 above. It follows that $S=S^\prime\cup \{s\}$ is an immediate parent of $S^\prime$ as specified by Rule 1 to compute immediate parents as $s$ is a maximal set independent of $S^\prime$ by hypothesis.

Let $S^\prime
=(S \setminus \{s\}) \cup\{\sigma_1,\ldots \sigma_k \}$ with $k \geq 1$ be a set as defined by Rule 2 above. It follows that $S=S^\prime \setminus \{\sigma_1,\ldots, \sigma_k\}\cup\{s\}$ is an immediate parent of $S'$ defined by Rule 2 to compute immediate parents. In fact, since $\forall i=1,\ldots, k, \sigma_i=s \cup \{l_i\}$ is independent of $S \setminus \{s\}$, with $l_i \in \{1,\ldots, p\}$, $s$ is by construction contained in all $\sigma_i$, and is a maximal set dominated by $S^\prime$. Finally, it is not contained in any maximal set independent of $S^\prime$ because all possible extensions of $s$ are comparable to some $s\cup\{l_i\} = \sigma_i$ already present in $S^\prime$.

Lastly, let $S^\prime=(S\setminus \{s_i,s_j\})\cup\{\sigma\}$ be a child of $S$ defined by Rule 3. Then $S=S^\prime\setminus\{\sigma\}\cup\{s_i,s_j\}$ is an immediate parent of $S^\prime$ defined by Rule 3 to immediate compute parents. 
In fact, $s_i$ and $s_j$ are maximal sets dominated by $S^\prime$ because they are maximally dominated by $\sigma$ and no other set in $S^\prime$ contains them. They are not contained in any maximal set independent of $S^\prime$ because by hypothesis $s_i$ is already a maximal set independent of $S\setminus\{s_i\}$ and $s_j$ is already a maximal independent set of $S\setminus\{s_j\}$. Furthermore, $s_i$ and $s_j$ are contained in $\sigma$ by construction and as $S\setminus\{s_i\}$ nor $S\setminus\{s_j\}$ yield covers of $\{1,\dots,p\}$ by hypothesis, neither $(S^\prime \setminus \{\sigma\}) \cup \{s_1\}$ nor $(S^\prime \setminus \{\sigma\}) \cup \{s_2\}$ yield covers of $\{1,\dots,p\}$, proving our statement that $S=S^\prime\setminus\{\sigma\}\cup\{s_i,s_j\}$ is an immediate parent of $S^\prime$ defined by Rule 3 to compute immediate parents.
\end{proof}

\begin{Theorem}\label{theorem:children}
$S'$ is an immediate child of $S$ in $\mathcal{S}_p$ if and only if $S'$ is generated by one of the 3 rules to compute children presented above.
\end{Theorem}

\begin{proof}
    The proof follows from Lemma \ref{lemma:duality}, which states that if set $S'$ is defined by one of the rules to compute children from a given set $S$, then $S'$ is a child of $S$, and $S$ is an immediate parent of $S'$.  By Theorem \ref{theorem:parents} there are no other possible forms of children and $\nexists S^{\prime\prime}\in \mathcal{S}_p$ such that $S^\prime\preccurlyeq S^{\prime\prime} \preccurlyeq S$, therefore making $S'$ an immediate child of $S$ in $\mathcal{S}_p$.
\end{proof}


\subsection{Cardinality of the True set of an immediate child}
\label{sec:properties-child}

Section~\ref{sec:properties-parent} shows that the cardinality difference between the True set of an immediate parent of a function $S$ and the True set of $S$ is either one or two (depending on the rules used to generate the immediate parent). This section relies on this result and Theorems \ref{theorem:parents} and \ref{theorem:children} to establish the converse result for immediate children.

\begin{Proposition}\label{prop:rule1-child}If $S'$ is an immediate child of $S$ generated by Rule 1, then $|\mathbb{T}(S')| = |\mathbb{T}(S)|-1$.
\end{Proposition}
\begin{proof}
    The proof follows immediately from Lemma \ref{lemma:duality} and Proposition \ref{prop:rule1-parent} because $S$ is an immediate parent of $S^\prime$ generated by Rule 1 to compute immediate parents, therefore $|\TT(S)|=|\TT(S^\prime)|+1$.
\end{proof}

\begin{Proposition}\label{prop:rule2-child}If $S'$ is an immediate child of $S$ generated by Rule 2, then $|\mathbb{T}(S')| = |\mathbb{T}(S)|-1$.
\end{Proposition}
\begin{proof}
    The proof follows immediately from Lemma \ref{lemma:duality} and Proposition \ref{prop:rule2-parent} because  $S$ is an immediate parent of $S^\prime$ generated by Rule 2 to compute immediate parents, therefore $|\TT(S)|=|\TT(S^\prime)|+1$.
\end{proof}

\begin{Proposition}\label{prop:rule3-child}If $S'$ is an immediate child of $S$ generated by Rule 3, then $|\mathbb{T}(S')| = |\mathbb{T}(S)|-2$.
\end{Proposition}
\begin{proof}
    The proof follows immediately from Lemma \ref{lemma:duality} and Proposition \ref{prop:rule3-parent} because  $S$ is an immediate parent of $S^\prime$ generated by Rule 3 to compute immediate parents, therefore $|\TT(S)|=|\TT(S^\prime)|+2$.
\end{proof}

\section{\label{sec:implementation}Implementation}

Following the rules presented in the previous sections, we developed two algorithms to compute the set of immediate parents and immediate children.
These algorithms receive as input a Boolean function $f$ in its set-representation $S \in \mathcal{S}_p$ and are made available as a dedicated Python library in \url{https://github.com/ptgm/pyfunctionhood}, under the GNU General Public License v3.0 (GPL-3.0).

Algorithm~\ref{alg:parents} shows the pseudo-code to compute the immediate parents (a step-by-step description of both algorithms is available in Appendix~\ref{sec:app:algorithms} as supplementary material). We start by computing the set $\mathcal{C}$ of maximal sets independent of $S$ (line~\ref{line:p:indep}), which in the worst case is $O(2^p)$ time complexity, dominating the overall running time. However, this is much more effective than generating and filtering the whole Hasse diagram which is $O(2^{2^p})$. Next, we iterate over the maximal sets dominated by $S$, whose number is also combinatorial. Hence, computing a Boolean function's set of immediate parents is $O(2^p)$.

{\footnotesize \begin{algorithm}[h!]
    \caption{Algorithm to Compute Immediate Parents}
    \label{alg:parents}
    \footnotesize
    \textbf{Input:} An element $S \in \mathcal{S}_p$\\
    \textbf{Output:} A set $\mathcal{P} \subseteq \mathcal{S}_p$ with all immediate parents of $S$ in $\mathcal{S}_p$\\
    \begin{algorithmic}[1] 
        \STATE {\em getMaxInd($S$)} : set of all maximal sets independent of $S$ 
        \STATE {\em getMaxDom($S$)} : set of all maximal sets dominated by $S$
        \STATE {\em getContaining($d,S$)} : set of all sets in $S$ containing $d$
        \STATE {\em isCover($S$)} : True if $S$ is a cover of $\{1,\dots, p\}$

        \STATE $\mathcal{P} \longleftarrow \varnothing $
        \STATE $\mathcal{C}\longleftarrow getMaxInd(S)$ \label{line:p:indep}

        \FOR{$c \in \mathcal{C}$} \label{line:p:for_indep}
            \STATE $\mathcal{P} \longleftarrow \mathcal{P} \cup \{ S \cup \{c\}\}$ \label{line:p:for_indep_end} \hfill // Rule 1
        \ENDFOR
        \STATE $\mathcal{D} \longleftarrow getMaxDom(S)$ \label{line:p:maxdom}
        \FOR{$d \in \mathcal{D}$} \label{line:p:maxdom_iter}
            \IF{$|getContaining(d,\mathcal{C})|>0$}
                \STATE $\mathcal{D} \longleftarrow \mathcal{D}\setminus\{d\}$
            \ENDIF
        \ENDFOR \label{line:p:maxdom_iter_end}
        \STATE $\mathcal{D}notused \longleftarrow \emptyset$ \label{line:p:map}

        \FOR{$d \in \mathcal{D}$} \label{line:p:iterD}
            \STATE $Containing \longleftarrow getContaining(d,S)$ 
            \IF{$isCover(S\setminus Containing \cup \{d\})$}
                \STATE $\mathcal{P} \longleftarrow \mathcal{P} \cup \{S\setminus Containing \cup \{d\}\}$ \hfill // Rule 2
            \ELSE
                \FOR{$s \in Containing$}
                    \STATE $\mathcal{D}notused[s] \longleftarrow \mathcal{D}notused[s] \cup \{d\}$ 
                \ENDFOR
            \ENDIF
        \ENDFOR \label{line:p:iterD_end}
        
        \FOR{$s \in \mathcal{D}notused$} \label{line:p:iterNotUsed}
            \FOR{$i=0 \to |\mathcal{D}notused[s]|-2$} \label{line:p:innerLoop}
                \STATE $d_i \longleftarrow \mathcal{D}notused[s][i]$
                \FOR{$j=i+1 \to |\mathcal{D}notused[s]|-1$}
                    \STATE $d_j \longleftarrow \mathcal{D}notused[s][j]$
                    \STATE $\mathcal{P} \longleftarrow \mathcal{P} \cup \{S\setminus \{s\} \cup \{d_i,d_j\}\}$ \hfill // Rule 3
                \ENDFOR
            \ENDFOR \label{line:p:innerLoop_end}
        \ENDFOR \label{line:p:iterNotUsed_end}
        \RETURN $\mathcal{P}$
    \end{algorithmic}
\end{algorithm}
}

To illustrate the complexity of the problem, we computed 100 random walks starting at the infimum (supremum) to the supremum (infimum) functions, considering an increasing number of variables (from 2 to 11). For a given trace, all neighbouring parents (children) of a given function are computed, and one of these parents (children) is randomly and uniformly chosen to be the next function to be part of the trace. This is repeated until the trace reaches the supremum (infimum) function.
At each step, for statistics, we keep the number of parents (children) generated by each of the three Rules, as well as the size and time of each trace.

Figure \ref{fig:results_hist} shows in the first Y-axis the average (over 100 traces) cumulative number of generated parents (or children) per rule. Interestingly, we observe that the number of functions generated by Rule 1 is exponential. 
Also, we observe that the number of functions generated by Rule 3 is not exponential and that proportionally to the other rules, is much smaller for higher dimensions.
We plot on the second Y-axis of Figure \ref{fig:results_hist} the average number of parents (in solid lines) and children (in dashed lines) generated per rule, by dividing the average number of functions per rule over the 100 traces by the average trace size.
We observe that indeed, not only there are more functions generated by Rule 1 due to the exponential increase in the trace size, but also that each function generates more neighbouring parents (or children) with increasing dimensions.
Also, we confirm that the cumulative number of functions generated by Rule 3 is not exponential. We observe that this number, when divided by the trace size is close to 0, due to the exponential increase of the trace size.

\begin{figure}[t]
  \includegraphics[width=\linewidth]{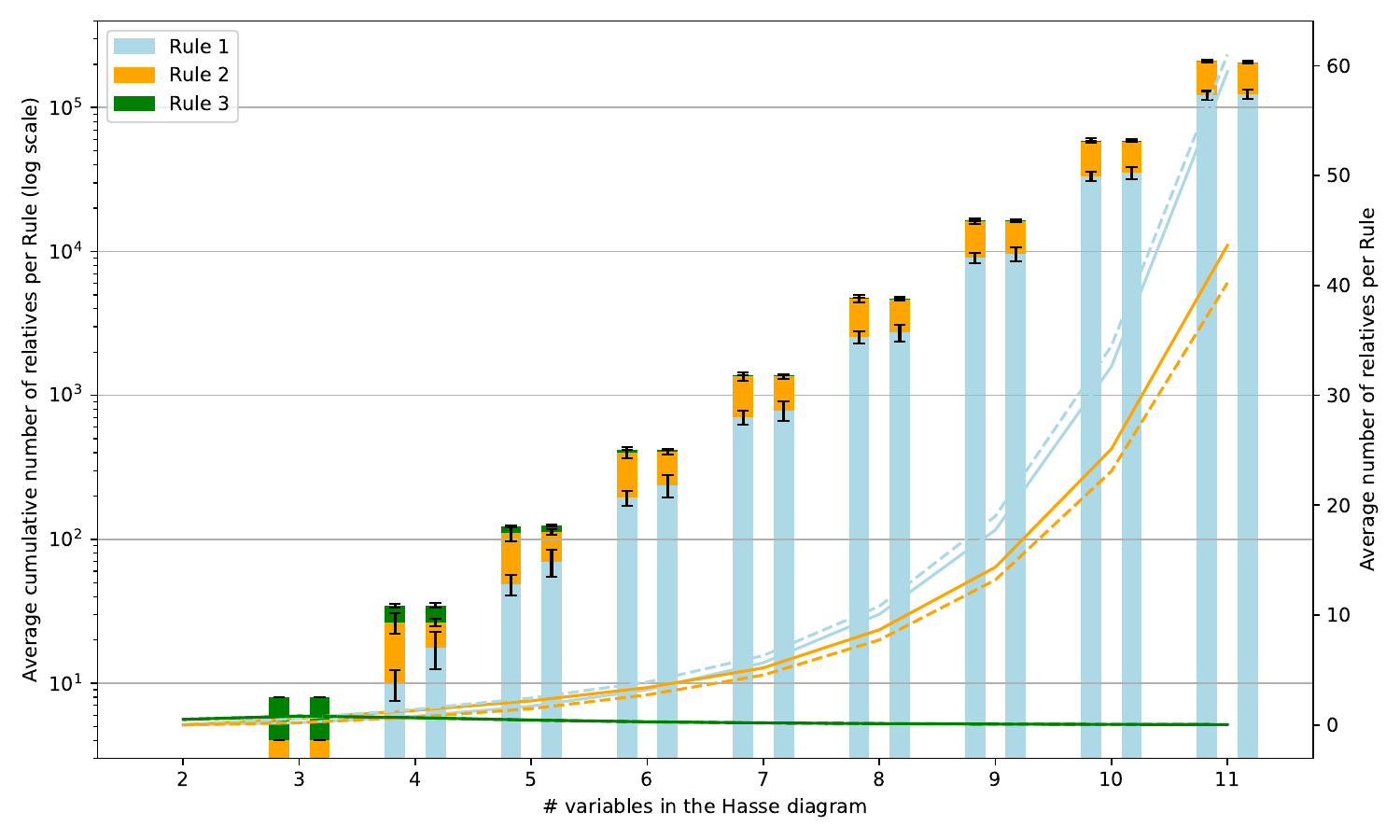}
  \caption{Average over 100 traces - from infimum to supremum, and vice-versa.
  Left Y-axis: cumulative number of generated parents and children per rule, left and right stacked histograms, respectively.
  Right Y-axis: number of generated parents (solid lines) and children (dashed lines) per rule, divided by average trace size.\label{fig:results_hist}}
\end{figure}

When analysing the performance of the proposed algorithms, we observe in the first Y-axis of Figure \ref{fig:results_time} that the time to compute a given trace increases exponentially with increasing dimensions, with the main contributing factor being the exponential increase of the trace size, which is shown in the second Y-axis.
However, the time to compute a given trace is not only dependent on the trace size but also on the number of parents (or children) generated per function (see Figure~\ref{fig:results_hist}), as we can observe that the time increases slightly more than linear in logscale.
When comparing the trace times when generating parents against the trace times when generating children, we can observe that the average time to generate children is slightly higher, suggesting potential improvements.

\begin{figure}[h]
    \includegraphics[width=\linewidth]{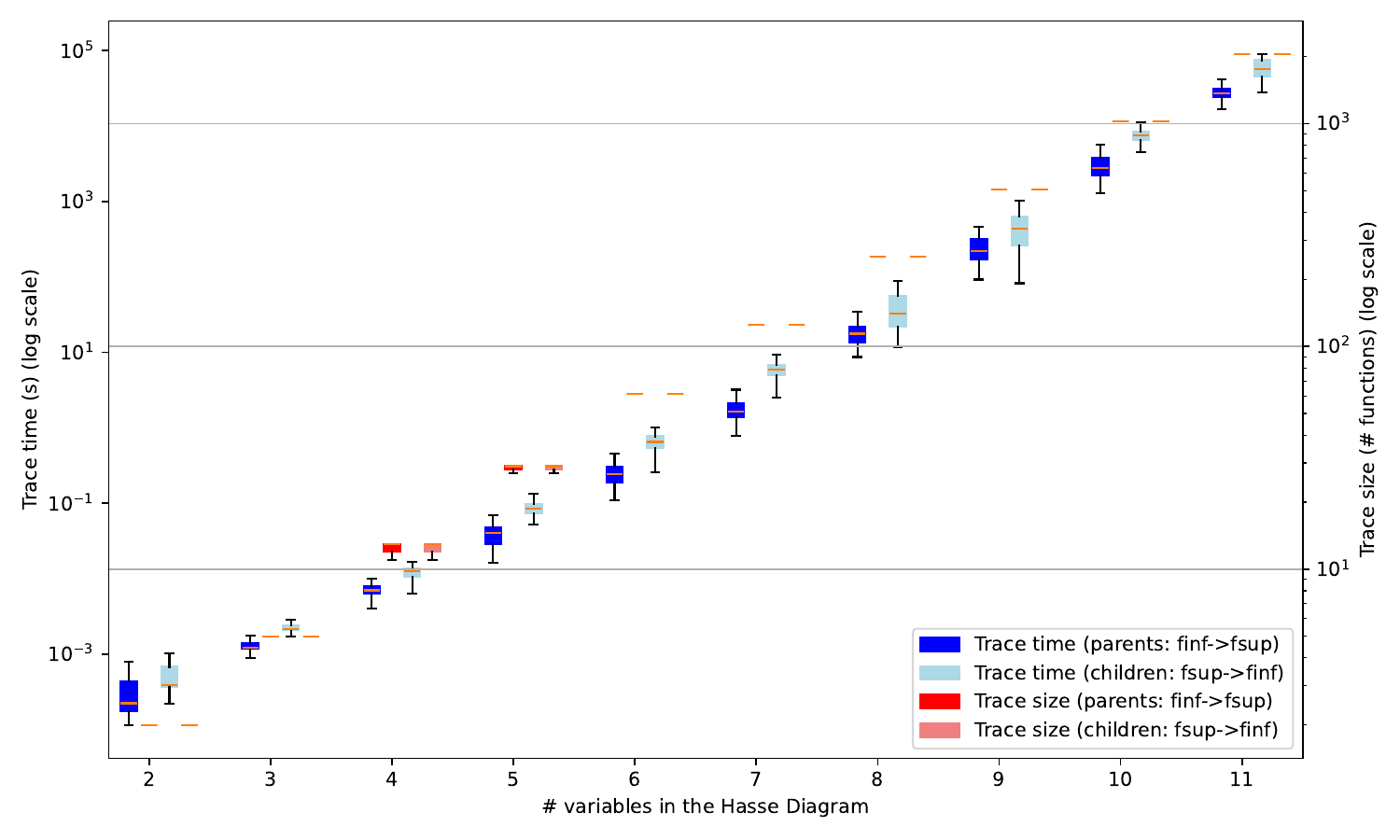}
    \caption{Average over 100 traces - from infimum to supremum and vice-versa.
    Box plot distribution generating parents (left) and generating children (right) along a trace.\label{fig:results_time}}
\end{figure}

\section{\label{sec:conclusion}Conclusion and prospects}

When defining Boolean models of regulatory networks, the choice of regulatory functions that ensures the desired dynamics is inherently hard due to the lack of regulatory data. In this work, we have characterised the set of monotone Boolean functions complying with a fixed topology of a regulatory network. 
In particular, we have specified its cardinality and its structure as a Partial Ordered set (poset).
Exploiting the poset structure, we re-defined the set of rules proposed in Cury {\it et al.} \cite{cury2019} to compute the direct neighbours of any monotone Boolean function.
These rules permit to navigate locally in the function space without having to generate the whole set of functions and subsequently compare them, which would unnecessarily use memory and CPU resources. We have assessed the number of states added or removed in the True set of a function when considering one of its direct neighbours. Such results can inform about the impact on the model dynamics in terms of transitions when the regulatory function of a component is modified. 

A dedicated Python library is freely available, under the GNU General Public License v3.0 (GPL-3.0), implementing both the three rules to compute the immediate parents given a reference monotone non-degenerate Boolean function, as well as the three rules to compute its immediate children.
It can be used in three distinct manners: as a library integrated in other tools, in the command line passing the reference function as an argument, or using a graphical interface developed with Tkinter. A small tutorial is presented in \url{https://github.com/ptgm/pyfunctionhood}.

As prospects, the rules to obtain neighbouring regulatory functions turn out to be useful for Probabilistic Boolean Networks (PBNs) as introduced by Shmulevich {\it et al.} \cite{shmulevich2002}. In contrast with Boolean models in which each component is associated with a unique regulatory function, PBNs introduce uncertainty in the regulatory functions governing the behaviours of model components. This is done by providing, for each component, a set of regulatory functions, each with a given probability.
We propose to associate a reference regulatory function with a certain probability and to distribute the remaining probability to the neighbouring functions (possibly at varying distances). This could be used to study the robustness of a given reference function with respect to the desired observations.

Furthermore, when a model does not meet specific requirements, the knowledge of the direct neighbourhood of regulatory functions could allow to perform local searches to improve model outcomes, with a minimal impact on the regulatory structure.
Additionally, it would allow for the qualification of the set of models complying with certain requirements, such as: models that have the same regulatory network, but different functions; or models capable of satisfying similar dynamical restrictions.
This has a huge impact on tools like ModRev~\cite{gouveia2020}, which proposes minimal repairs that are as close as possible to the original reference functions.

Finally, numerous tools are provided and integrated in the context of the CoLoMoTo (Consortium for Logical Models and Tools) at \url{https://colomoto.github.io} \cite{colomoto2018}.
This new library will be made available as part of these tools, as well as potentially being integrated in existing ones.

\subsubsection*{Funding}
 JC acknowledges the support from the Brazilian agency CAPES, with a one-year research fellowship to visit IGC. This work has been further supported by the Portuguese national agency Fundação para a Ciência e a Tecnologia (FCT) with references PTDC/EEI-CTP/2914/2014, DOI:10.54499/2023.14280.PEX,
DOI:10.54499/2024.07475.IACDC, and 
DOI:10.54499/UIDB/50021/2020.

\bibliographystyle{abbrv}
\bibliography{biblio}

\appendix
\section*{Appendix}

\section{Partially Ordered Sets}\label{app:posets}

Given a set $S$, a \emph{Partial Order} on $S$ is a binary relation $\preceq$ on $S$ that is  reflexive, antisymmetric and transitive.

The pair $(S,\preceq)$ defines a \emph{Partially Ordered Set} (\emph{poset}) in which two  elements $s,s^\prime$ of $S$ are said  \emph{comparable} if either $s \preceq s^\prime$ or $s^\prime \preceq s$.

A finite poset $(S,\preceq)$ can be graphically represented as a \emph{Hasse Diagram} ({HD}), where each element of $S$ is a vertex in the plane, and an edge connects a vertex $s\in S$ to a vertex $s^\prime \in S$ placed above iff: $s \prec s^\prime, \textrm{ and } \nexists s^{\prime\prime} \in S \textrm{ such that } s \prec s^{\prime\prime} \prec s^\prime$. 

Given $A \subseteq S$,  $u \in S$  is an \emph{upper bound} (resp.  \emph{lower bound}) of $A$ in the poset $(S,\preceq)$ if $s\preceq u$ (resp. $l \preceq s$) for all $s \in A$.  A least upper bound (resp. greatest lower bound) of $A$ is called a supremum (resp. an infimum) of $A$. $(S,\preceq)$ is \emph{bounded} if $S$ has both an infimum and a supremum. 

A \emph{chain} in a poset $(S,\preceq)$ is a subset of $S$ in which all the elements are pairwise comparable. The counterpart notion is an \emph{antichain}, defined as a subset of $S$ in which any two elements are incomparable.

Furthermore, an element $s \in S$ is \emph{independent} of an antichain $A\subsetneq S$ if $A \cup \{s\}$ remains an antichain, namely, $s$ is incomparable to any element of $A$. 

Let us now consider the specific case of $2^S$, the set of all subsets of a set $S$. Obviously $(2^S,\subseteq)$ defines a poset.  A set of elements of $2^S$ whose union contains $S$ is called a \emph{cover} of $S$. 

Given $A\subset 2^S$, $s\in (2^S\setminus A)$  \emph{is dominated by} $A$ if $\exists a\in A $ such that $ s\subsetneq a$ ({\it i.e,} $s$ is contained by at least one element of $A$). Furthermore, $s\in 2^S\setminus A$ is a \emph{ maximal set dominated} by $A$ if $\forall s'$ dominated by $A$, $s'\not\subset s$.
Finally, $s\in 2^S\setminus A$ is a \emph{maximal set independent} of $A$ if $s$ is independent of $A$, and $\forall x\in S, x\not\in s, \exists a\in A$ such that $a\subset (s\cup\{x\})$.

Figure \ref{fig:HD1234} illustrates the different notions introduced above on the HD of the poset  $(2^{\{1,2,3,4\}},\subseteq)$.

\begin{figure}[h]
	\begin{center}
\resizebox{0.6\textwidth}{!}{ 
	\begin{tikzpicture}
		\tikzstyle{term} = [draw,shape=rectangle,minimum height=2em,text centered, fill, color=gray!20]
		\tikzstyle{term2} = [draw,shape=rectangle,minimum height=2em,text centered, fill, color=orange]
		\tikzstyle{term3} = [draw,shape=rectangle,minimum height=2em,text centered, fill, color=green!50]
		\tikzstyle{term4} = [draw,shape=rectangle,minimum height=2em,text centered, fill, color=violet!50]
		\tikzstyle{term5} = [draw,shape=rectangle,minimum height=2em,text centered, fill, color=blue!50]
		\tikzstyle{term6} = [draw,shape=ellipse,minimum height=2em,text centered, fill, color=green!50]
		\tikzstyle{term7} = [draw,shape=ellipse,minimum height=2em,text centered, fill, color=violet!50]
		\tikzstyle{link} = [-]
		\def\xspace{2}
		\def\yspace{-2}
		
		\node[term5] (1234) at (0,0){\textcolor{black}{$\{1,2,3,4\}$}};
		\node[term2](123) at (-2.5*\xspace,\yspace) {\textcolor{black}{$\{1,2,3\}$}};
		\node[term5](124) at (-1*\xspace,\yspace) {\textcolor{black}{$\{1,2,4\}$}};
		\node[term6](134) at (1*\xspace,\yspace) {\textcolor{black}{$\{1,3,4\}$}};
		\node[term5](234) at (2.5*\xspace,\yspace) {\textcolor{black}{$\{2,3,4\}$}};
	\node[term7](12) at (-4*\xspace,2*\yspace) {\textcolor{black}{$\{1,2\}$}};	
	\node[term7](13) at (-2.5*\xspace,2*\yspace) {\textcolor{black}{$\{1,3\}$}};
	\node[term3](14) at (-1*\xspace,2*\yspace) {\textcolor{black}{$\{1,4\}$}};
	\node[term7](23) at (1*\xspace,2*\yspace) {\textcolor{black}{$\{2,3\}$}};
	\node[term2](24) at (2.5*\xspace,2*\yspace) {\textcolor{black}{$\{2,4\}$}};	
	\node[term3](34) at (4*\xspace,2*\yspace) {\textcolor{black}{$\{3,4\}$}};	
	\node[term4](1) at (-2.5*\xspace,3*\yspace) {\textcolor{black}{$\{1\}$}};
	\node[term4](2) at (-1*\xspace,3*\yspace) {\textcolor{black}{$\{2\}$}};
	\node[term4](3) at (1*\xspace,3*\yspace) {\textcolor{black}{$\{3\}$}};
	\node[term4](4) at (2.5*\xspace,3*\yspace) {\textcolor{black}{$\{4\}$}};
	\node[term4](0) at (0,4*\yspace) {\textcolor{black}{$\{\varnothing\}$}};
	
\draw[link] (1234) edge (123);
\draw[link] (1234) edge (124);
\draw[link] (1234) edge (134);	
\draw[link] (1234) edge (234);	
\draw[link] (123) edge (12);
\draw[link] (123) edge (13);
\draw[link] (123) edge (23);
\draw[link] (124) edge (12);
\draw[link] (124) edge (14);
\draw[link] (124) edge (24);	
\draw[link] (134) edge (13);
\draw[link] (134) edge (14);
\draw[link] (134) edge (34);
\draw[link] (234) edge (23);
\draw[link] (234) edge (34);
\draw[link] (234) edge (24);	
\draw[link] (1) edge (12);
\draw[link] (1) edge (13);
\draw[link] (1) edge (14);
\draw[link] (2) edge (12);
\draw[link] (2) edge (23);
\draw[link] (2) edge (24);
\draw[link] (3) edge (13);
\draw[link] (3) edge (23);
\draw[link] (3) edge (34);
\draw[link] (4) edge (14);
\draw[link] (4) edge (24);
\draw[link] (4) edge (34);
\draw[link] (0) edge (1);
\draw[link] (0) edge (2);
\draw[link] (0) edge (3);
\draw[link] (0) edge (4);
\end{tikzpicture}}
\end{center}
\caption{Hasse diagram of the poset $(2^{\{1,2,3,4\}},\subseteq)$. It is clearly bounded ($\varnothing$ being its infimum, and $\{1,2,3,4\}$ its supremum). The orange vertices define an anti-chain $A=\{\{1,2,3\},\{2,4\}\}$, which is a cover of $\{1,2,3,4\}$; green vertices indicate the sets independent of $A$ (with $\{1,3,4\}$ being maximal); blue vertices dominate $A$, and violet vertices are dominated by $A$ (with $\{1,2\}, \{1,3\}, \{2,3\}$ being maximal). \label{fig:HD1234}}
\end{figure}

\section{Proof of Theorem \ref{theorem:parents}}\label{app:proof_th1}

\begin{proof}
Notice that the sets $S$ and $S'$ in the three rules satisfy $S\preccurlyeq S'$. Moreover, $S'$ defined in Rules 1 and 2 clearly yields a cover of $\{1, \dots, p\}$. This is also the case in Rule 3 because $\sigma_1$ and $\sigma_2$ are maximal sets dominated by $S$ contained in $s$ thus: $\exists i\in\{1,\ldots, p\}, i\not\in \sigma_2: \{i\}\cup\sigma_1=s$ and $\exists j\in\{1,\ldots, p\}, j\not\in \sigma_1: \{j\}\cup\sigma_2=s$ thus $S'=S\setminus \{s\}\cup\{\sigma_1,\sigma_2\}$ is a cover of $\{1,\ldots,p\}$.

Hence, for the 3 rules, $S'$ is a {\it valid parent} of $S$ in $\mathcal{S}_p$ when $S'$ can be defined, {\em i.e.,} when each rule can be applied.
\medskip

Let $S'=S \cup \{\sigma\}$ be a set generated by Rule 1. To show that it is an \emph{immediate parent} of $S$, let us assume by contradiction that: $\exists S^{\prime\prime} \in \mathcal{S}_p$ such that $S \preceq S^{\prime\prime} \preceq S^\prime$. Then, by definition, 
\begin{align*}
\begin{cases}
  \forall s\in S, \exists s^{\prime\prime}\in S^{\prime\prime} : s^{\prime\prime} \subseteq s, & (1)\\
  \forall s^{\prime\prime}\in S^{\prime\prime}, \exists s^\prime\in S^\prime : s^{\prime}\subseteq s^{\prime\prime}. & (2)
\end{cases}
\end{align*}

One of two things can happen, either $S^{\prime\prime}$ is of the form: (a) $S^{\prime\prime}=S \cup \{s^{\prime\prime}_1, \dots, s^{\prime\prime}_{l} \}$ with $l \geq 1$ and $\forall i=1,\ldots,l$, $s^{\prime\prime}_i$ independent of $S$; or (b) $S \not\subseteq S^{\prime\prime}$. 

The first case (a) implies that $\sigma$ is contained in all the sets $s^{\prime\prime}_i,\,i=1,\dots,l$, a contradiction with the fact that $\sigma$ is a maximal set independent of $S$.
In the second case (b), there must be a set $s \in S$, such that $s \not\in S^{\prime\prime}$. However, $(1)$ implies that there must exist $s^{\prime\prime} \in S^{\prime\prime}$ such that $s^{\prime\prime} \subset s$.
Condition $(2)$ enforces this $s^{\prime\prime}$ to have a witness in $S^\prime$, which must necessarily be $\sigma$.
However this entails $\sigma\subseteq s^{\prime\prime}\subseteq s$, a contradiction with $\sigma$ being a maximal set independent of $S$.
Therefore sets generated by Rule 1 are immediate parents of $S$ in $\mathcal{S}_p$.
\medskip

Proving the immediacy of parents generated by Rule 2 and Rule 3 is less straightforward. Note that by construction, parents generated by Rule 1 can never lie between $S$ and a parent of $S$ generated by Rule 2 or Rule 3.
A stronger statement can be made: if $\exists S^{\prime\prime} \in \mathcal{S}_p$ such that $S \preccurlyeq S^{\prime\prime} \preccurlyeq S^{\prime}$, where $S^{\prime}$ is a parent generated by Rule 2 or Rule 3, then $S\not\subseteq S^{\prime\prime}$. 

To show this, suppose towards a contradiction that $S^\prime$ is a parent generated by Rule 2 such that $S\preccurlyeq S^{\prime\prime} \preccurlyeq S^{\prime}$ where $S^{\prime\prime}$ contains $S$. As $S \subseteq S^{\prime\prime}$, $S^{\prime\prime}$ is of the form $S^{\prime\prime}= S \cup \{s^{\prime\prime}_1, \dots,s^{\prime\prime}_{l}\} $ for some $l \geq 1$. Each set in $S^{\prime\prime}$ needs its witness in $S^\prime$ to ensure $S^{\prime\prime} \preccurlyeq S^\prime$. As $S \subseteq S^{\prime\prime}$ and $S \preccurlyeq S^\prime$, each set $s$ shared between $S$ and $S^{\prime\prime}$ already has its witness in $S^\prime$. As for sets $s^{\prime\prime}_1, \dots,s^{\prime\prime}_{l}$, they can only have $\sigma$ as their witness, otherwise some $s^{\prime\prime}_i$ would not be independent of $S$. This means that $\sigma \subsetneq s^{\prime\prime}_{i}$ for all $i=1,\dots, l$, a contradiction with $\sigma$ not being contained in any maximal set independent of $S$.
The result for parents generated by Rule 3 follows similarly. The only possible witnesses for any set in $S^{\prime\prime}$ not in $S$ are $\sigma_1$ and $\sigma_2$, meaning at least one of them must be contained in some maximal set independent of $S$, a contradiction with condition (b) of Rule 3.


Let us now prove the immediacy of parents generated by Rule 2. Let $S^\prime=(S \setminus \{s_1 , \dots, s_{k}\} ) \cup \{\sigma\}$ be a parent of $S$ generated by Rule 2 and suppose $S^\prime$ is not an immediate parent of $S$: $\exists S^{\prime\prime} \in \mathcal{S}_p$ such that $S \preccurlyeq S^{\prime\prime} \preccurlyeq S^{\prime}$. Conditions $(1)$ and $(2)$ above apply again. By our previous remark, $S\not\subseteq S^{\prime\prime}$ and so $\exists s \in S$ such that $s \not\in S^{\prime\prime}$. By $(1)$, this set $s$ needs a witness in $S^{\prime\prime}$ which we denote by $s^{\prime\prime}_1$ ($s^{\prime\prime}_1 \subsetneq s$). 
In turn, this $s^{\prime\prime}_1$ requires a witness in $S^\prime$, which can only be $\sigma$ because all sets in $S^\prime$ other than $\sigma$ are in $S$. Therefore $\sigma\subseteq s^{\prime\prime}_1 \subsetneq s$. By hypothesis, $\sigma$ is a maximal set dominated by $S$, thus if $\sigma \subsetneq s$, and $|s|=q$, then $|\sigma|=q-1$. This enforces $|s^{\prime\prime}_1|=q-1$ and hence $s^{\prime\prime}_1$ must coincide with $\sigma$ ($s^{\prime\prime}_1=\sigma$, thus $\sigma \in S^{\prime\prime}$). 

As $S^{\prime\prime} \neq S^\prime$, there must be some other set $s^{\prime\prime}_2$ in $S^{\prime\prime}$, such that $s^{\prime\prime}_2 \not\in S^\prime$.

Again by condition $(2)$, $s^{\prime\prime}_2$ requires a witness in $S^\prime$. This witness can only be some set $s$ belonging to both $S$ and $S^\prime$. That is $\exists s \in S, s\in S^\prime : s \subseteq s^{\prime\prime}_{2}$, which results in a contradiction since this same $s\in S$ would no longer have its witness in $S^{\prime\prime}$, failing to meet condition $(1)$. This proves that sets generated by Rule 2 are immediate parents of $S$ in $\mathcal{S}_p$.

\medskip

To prove that Rule 3 generates immediate parents, let $S'=(S \setminus \{s\}) \cup \{\sigma_{1} , \sigma_{2} \}$ be a parent of $S$ generated by Rule 3 and assume towards a contradiction that $\exists S^{\prime\prime} \in \mathcal{S}_p$ such that $S \preccurlyeq S^{\prime\prime} \preccurlyeq S^\prime$. Since $S \not\subseteq S^{\prime\prime}$, we have that $\exists s_i \in S: s_i \not\in S^{\prime\prime}$. Let $s^{\prime\prime}_1$ be a witness for $s_i$ in $S^{\prime\prime}$ ({\it i.e.} $s^{\prime\prime}_1 \subsetneq s_i$). Condition $(2)$ establishes that $s^{\prime\prime}_1$ requires a witness in $S'$, which can only be $\sigma_1$ or $\sigma_2$ since all other sets in $S^\prime$ are also elements of $S$. Without loss of generality assume that $s^{\prime\prime}_1$ is this witness: $ \sigma_1\subseteq s^{\prime\prime}_1\subsetneq s$.

Given that by hypothesis $\sigma_{1}$ is maximal dominated by $S$ and $\sigma_{1} \subsetneq s_i$, if $|s_i|=q$ then $|\sigma_1|=q-1$. This enforces $|s^{\prime\prime}_1|=q-1$, and thus $s^{\prime\prime}_1$ must coincide with $\sigma_{1}$ ($ s^{\prime\prime}_1 = \sigma_{1}$ and hence $\sigma_{1} \in S^{\prime\prime}$).

In order to have $S^{\prime\prime} \neq S'$, there must then exist some other set $s^{\prime\prime}_{2}\in S^{\prime\prime}$, such that $s^{\prime\prime}_{2} \not\in S'$ (thus $s^{\prime\prime}_{2}\neq \sigma_1$ and $s^{\prime\prime}_{2}\neq \sigma_2$). This set $s^{\prime\prime}_{2}$ cannot be $s$, the set removed from $S$, since we have just concluded that $\sigma_1 \in S^{\prime\prime}$ and given that $\sigma_{1} \subseteq s$ by Rule 3 condition (c), this would entail that $S^{\prime\prime}$ contained comparable sets.
The end of the proof now follows similarly to that of Rule 2. This $s^{\prime\prime}_{2}$ needs a witness in $S'$, and this witness can only be a common set of $S$ and $S'$, which leads to a contradiction by leaving this witness without its own witness in $S^{\prime\prime}$. This proves that sets generated by Rule 3 are indeed immediate parents of $S$ in $\mathcal{S}_p$.

\medskip
Lastly, we prove that no other possible immediate parents of $S$ in $\mathcal{S}_p$ exist.

Let us consider $S^{\prime}$ an immediate parent of $S$ in $\mathcal{S}_p$ ($S \preccurlyeq S^{\prime}$), either $S$ is integrally contained in $S^{\prime}$, that is every set $s\in S$ is also in $S^{\prime}$, or there is at least one set in $S$ that is not in $S^{\prime}$. Towards a contradiction, let us suppose that $S^{\prime}$ is an immediate parent of $S$ not generated by any of our three rules.

If $S$ is integrally contained in $S^{\prime}$ then $S^{\prime}=S \cup \{s^{\prime}_1, \dots, s^{\prime}_{l}\}$ for some $l\geq 1$. Given that by hypothesis, $S^{\prime}$ is not generated by any of the 3 rules, either $k\geq 2$ or $k=1$ and $s^{\prime}_1$ is not a maximal set independent of $S$. 

If $k\geq 2$ then removing any set in $\{s^{\prime}_1, \dots, s^{\prime}_{l}\}$ yields a child of $S^{\prime}$ that is also a parent of $S$, hence $S^{\prime}$ is not an immediate parent of $S$.

If $k=1$ and $s^{\prime}_1$ is not a maximal set independent of $S$, then there exists $\sigma$, a maximal set independent of $S$ such that $s^{\prime}_1\subsetneq \sigma$. The set $S\cup\{\sigma\}$ obtained from $S^{\prime}=S\cup\{s^{\prime}_1\}$ by replacing $s^{\prime}_1$ by $\sigma$ yields a child of $S^{\prime}$ that is a parent of $S$, hence $S^{\prime}$ is once again not an immediate parent of $S$.

Now let us suppose that $S^{\prime}$ is an immediate parent of $S$ such that $S$ is not integrally contained in $S^{\prime}$. Let $s$ be a set of $S$ that is not in $S^{\prime}$, and let $s^{\prime}$ be its witness in $S^\prime$. This tells us that $S^{\prime}$ is either of the form $S^{\prime}= S \setminus \{s_1,\dots, s_k\} \cup \{s^{\prime}\}$ for some $k \geq 1$, or of the form $S^{\prime}= S \setminus \{s_1, \dots, s_k\} \cup \{s^\prime , s^{\prime}_1,s^{\prime}_2, \dots, s^{\prime}_l \}$ for some $k \geq 1$ and $l \geq 1$, where $\{s_1, \dots, s_k\}$ contains $s$, all sets in $S$ that contain $s^{\prime}$ and possibly other sets of $S$ as well.

If $S^{\prime}= S \setminus \{s_1,\dots, s_k\} \cup \{s^{\prime}\}$, then $s^{\prime}$ must be a witness for all sets $s_1, \dots, s_k$. As $S^{\prime}$ is not generated by any of the rules, either $s^{\prime}$ is not a maximal set dominated by $S$ or it is contained in some maximal set independent of $S$.
If $s^{\prime}$ is contained in a maximal set independent of $S$, $\sigma$, then $S \preccurlyeq S\cup\{\sigma\} \preccurlyeq S^{\prime}$, which leads to a contradiction with the assumption that $S^{\prime}$ is an immediate parent of $S$.
If $s^{\prime}$ is not a maximal set dominated by $S$, it is contained in a maximal set dominated by $S$, $\sigma$, and $S \preccurlyeq S \setminus \{s \in S : d \subseteq s\} \cup \{d\} \preccurlyeq S^{\prime\prime}$, which again leads to a contradiction with the assumption that $S^{\prime}$ is an immediate parent of $S$.

A similar reasoning applies if $S^{\prime}= S \setminus \{s_1, \dots, s_k\} \cup \{s^\prime , s^{\prime}_1,s^{\prime}_2, \dots, s^{\prime}_l \}$ for some $k \geq 1$ and $l \geq 1$. Let us suppose that $s^{\prime}$ is not a maximal set dominated by $S$ and let $\sigma$ be a maximal set dominated by $S$ that contains $s^{\prime}$. 
This implies: $ S \preccurlyeq S\setminus \bigl(\{s \in S : \sigma \subseteq s\}\bigr)\cup \{\sigma\} 
\cup \{s^{\prime}_i \in S^{\prime} : s^{\prime}_i \not\subseteq \sigma \} \preccurlyeq S^{\prime}$. This set laying between $S$ and $S^{\prime}$ 
is necessarily distinct from $S^{\prime}$ because it contains $\sigma$, and it is a cover of $\{1,\dots,p\}$ because it is endowed with the sets in $\{s^{\prime}_i \in S^{\prime} : s^{\prime}_i \not\subseteq \sigma \}$. Hence, we again have a contradiction with the assumption that $S^{\prime}$ is an immediate parent of $S$. 

Therefore, there is no other form of immediate parent than those defined by our 3 rules.
\end{proof}

\newpage
\section{\label{sec:app:examples}Illustration of proposed rules}
\begin{figure}[h!]
\begin{center}
\resizebox{0.85\textwidth}{!}{%
\begin{tabular}{|c|c|}\hline
\begin{minipage}{0.475\textwidth}\centering {\small
$\begin{array}{l@{\hspace{0cm}}r@{\hspace{0cm}}ll}
{\rm set}&\multicolumn{3}{l}{ S_1 = \{\{1,2,3\},\{3,4\}\}} \\
\multirow{2}{*}{parents}&\ldelim \{ {2}{3 mm} 
 & \{\{1,2,3\},\{1,2,4\},\{3,4\}\}&Rule1\\
& & \{\{1,3\},\{2,3\},\{3,4\}\}&Rule3\\
\end{array}$}
\end{minipage}
&
\begin{minipage}{0.475\textwidth}\centering \resizebox{\textwidth}{!}{
$\begin{array}{l@{\hspace{0.1cm}}r@{\hspace{0.1cm}}ll}
{\rm set}& \multicolumn{3}{l}{S_2 = \{\{1,2,3\},\{1,2,4\},\{3,4\}\}} \\
\multirow{3}{*}{children}&\ldelim \{ {3}{2 mm} 
 &\{\{1,2,4\},\{3,4\}\}&Rule1\\
& &\{\{1,2,3\},\{3,4\}\}&Rule1\\
& &\{\{1,2,3\},\{1,2,4\},\{1,3,4\},\{2,3,4\}\}&Rule2\\
\end{array}$}
\end{minipage}\\
\begin{minipage}{0.45\textwidth}
	\resizebox{\textwidth}{!}{ 
		\begin{tikzpicture}
			\tikzstyle{term} = [draw,shape=rectangle,minimum height=2em,text centered, fill, color=gray!20]
			\tikzstyle{term2} = [draw,shape=rectangle,minimum height=2em,text centered, fill, color=orange]
			\tikzstyle{term32} = [draw,ellipse,minimum height=2em,text centered, fill, color=green!50]
			\tikzstyle{term4} = [draw,shape=rectangle,minimum height=2em,text centered, fill, color=violet!50]
			\tikzstyle{term42} = [draw,ellipse,minimum height=2em,text centered, fill, color=violet!50]
			\tikzstyle{term5} = [draw,shape=rectangle,minimum height=2em,text centered, fill, color=blue!50]
			\tikzstyle{link} = [-]
			\def\xspace{1.5}
			\def\yspace{-1.5}
			
			\node[term] (1234) at (0,0){\textcolor{black}{$\{1,2,3,4\}$}};
			\node[term2](123) at (-2.5*\xspace,\yspace) {\textcolor{black}{$\{1,2,3\}$}};
			\node[term32](124) at (-1*\xspace,\yspace) {\textcolor{black}{$\{1,2,4\}$}};
			\node[term2](134) at (1*\xspace,\yspace) {\textcolor{black}{$\{1,3,4\}$}};
			\node[term2](234) at (2.5*\xspace,\yspace) {\textcolor{black}{$\{2,3,4\}$}};
			\node[term4](12) at (-4*\xspace,2*\yspace) {\textcolor{black}{$\{1,2\}$}};	
			\node[term42](13) at (-2.5*\xspace,2*\yspace) {\textcolor{black}{$\{1,3\}$}};
			\node[term4](14) at (-1*\xspace,2*\yspace) {\textcolor{black}{$\{1,4\}$}};
			\node[term42](23) at (1*\xspace,2*\yspace) {\textcolor{black}{$\{2,3\}$}};
			\node[term4](24) at (2.5*\xspace,2*\yspace) {\textcolor{black}{$\{2,4\}$}};	
			\node[term42](34) at (4*\xspace,2*\yspace) {\textcolor{black}{$\{3,4\}$}};	
			\node[term4](1) at (-2.5*\xspace,3*\yspace) {\textcolor{black}{$\{1\}$}};
			\node[term4](2) at (-1*\xspace,3*\yspace) {\textcolor{black}{$\{2\}$}};
			\node[term4](3) at (1*\xspace,3*\yspace) {\textcolor{black}{$\{3\}$}};
			\node[term4](4) at (2.5*\xspace,3*\yspace) {\textcolor{black}{$\{4\}$}};
			\node[term4](0) at (0,4*\yspace) {\textcolor{black}{$\{\emptyset\}$}};
			
			\draw[link] (1234) edge (123);
			\draw[link] (1234) edge (124);
			\draw[link] (1234) edge (134);	
			\draw[link] (1234) edge (234);	
			\draw[link] (123) edge (12);
			\draw[link] (123) edge (13);
			\draw[link] (123) edge (23);
			\draw[link] (124) edge (12);
			\draw[link] (124) edge (14);
			\draw[link] (124) edge (24);	
			\draw[link] (134) edge (13);
			\draw[link] (134) edge (14);
			\draw[link] (134) edge (34);
			\draw[link] (234) edge (23);
			\draw[link] (234) edge (34);
			\draw[link] (234) edge (24);	
			\draw[link] (1) edge (12);
			\draw[link] (1) edge (13);
			\draw[link] (1) edge (14);
			\draw[link] (2) edge (12);
			\draw[link] (2) edge (23);
			\draw[link] (2) edge (24);
			\draw[link] (3) edge (13);
			\draw[link] (3) edge (23);
			\draw[link] (3) edge (34);
			\draw[link] (4) edge (14);
			\draw[link] (4) edge (24);
			\draw[link] (4) edge (34);
			\draw[link] (0) edge (1);
			\draw[link] (0) edge (2);
			\draw[link] (0) edge (3);
			\draw[link] (0) edge (4);
	\end{tikzpicture}}
\end{minipage}& \begin{minipage}{0.45\textwidth}
	\resizebox{\textwidth}{!}{ 
		\begin{tikzpicture}
			\tikzstyle{term} = [draw,shape=rectangle,minimum height=2em,text centered, fill, color=gray!20]
			\tikzstyle{term2} = [draw,shape=rectangle,minimum height=2em,text centered, fill, color=orange]
			\tikzstyle{term3} = [draw,shape=rectangle,minimum height=2em,text centered, fill, color=green!50]
			\tikzstyle{term32} = [draw,ellipse,minimum height=2em,text centered, fill, color=green!50]
			\tikzstyle{term4} = [draw,shape=rectangle,minimum height=2em,text centered, fill, color=violet!50]
			\tikzstyle{term42} = [draw,ellipse,minimum height=2em,text centered, fill, color=violet!50]
			\tikzstyle{term5} = [draw,shape=rectangle,minimum height=2em,text centered, fill, color=blue!50]
			\tikzstyle{link} = [-]
			\def\xspace{1.5}
			\def\yspace{-1.5}
			
			\node[term] (1234) at (0,0){\textcolor{black}{$\{1,2,3,4\}$}};
			\node[term2](123) at (-2.5*\xspace,\yspace) {\textcolor{black}{$\{1,2,3\}$}};
			\node[term32](124) at (-1*\xspace,\yspace) {\textcolor{black}{$\{1,2,4\}$}};
			\node[term](134) at (1*\xspace,\yspace) {\textcolor{black}{$\{1,3,4\}$}};
			\node[term](234) at (2.5*\xspace,\yspace) {\textcolor{black}{$\{2,3,4\}$}};
			\node[term4](12) at (-4*\xspace,2*\yspace) {\textcolor{black}{$\{1,2\}$}};	
			\node[term42](13) at (-2.5*\xspace,2*\yspace) {\textcolor{black}{$\{1,3\}$}};
			\node[term3](14) at (-1*\xspace,2*\yspace) {\textcolor{black}{$\{1,4\}$}};
			\node[term42](23) at (1*\xspace,2*\yspace) {\textcolor{black}{$\{2,3\}$}};
			\node[term3](24) at (2.5*\xspace,2*\yspace) {\textcolor{black}{$\{2,4\}$}};	
			\node[term2](34) at (4*\xspace,2*\yspace) {\textcolor{black}{$\{3,4\}$}};	
			\node[term4](1) at (-2.5*\xspace,3*\yspace) {\textcolor{black}{$\{1\}$}};
			\node[term4](2) at (-1*\xspace,3*\yspace) {\textcolor{black}{$\{2\}$}};
			\node[term4](3) at (1*\xspace,3*\yspace) {\textcolor{black}{$\{3\}$}};
			\node[term4](4) at (2.5*\xspace,3*\yspace) {\textcolor{black}{$\{4\}$}};
			\node[term4](0) at (0,4*\yspace) {\textcolor{black}{$\{\emptyset\}$}};
			
			\draw[link] (1234) edge (123);
			\draw[link] (1234) edge (124);
			\draw[link] (1234) edge (134);	
			\draw[link] (1234) edge (234);	
			\draw[link] (123) edge (12);
			\draw[link] (123) edge (13);
			\draw[link] (123) edge (23);
			\draw[link] (124) edge (12);
			\draw[link] (124) edge (14);
			\draw[link] (124) edge (24);	
			\draw[link] (134) edge (13);
			\draw[link] (134) edge (14);
			\draw[link] (134) edge (34);
			\draw[link] (234) edge (23);
			\draw[link] (234) edge (34);
			\draw[link] (234) edge (24);	
			\draw[link] (1) edge (12);
			\draw[link] (1) edge (13);
			\draw[link] (1) edge (14);
			\draw[link] (2) edge (12);
			\draw[link] (2) edge (23);
			\draw[link] (2) edge (24);
			\draw[link] (3) edge (13);
			\draw[link] (3) edge (23);
			\draw[link] (3) edge (34);
			\draw[link] (4) edge (14);
			\draw[link] (4) edge (24);
			\draw[link] (4) edge (34);
			\draw[link] (0) edge (1);
			\draw[link] (0) edge (2);
			\draw[link] (0) edge (3);
			\draw[link] (0) edge (4);
	\end{tikzpicture}}
\end{minipage}
\\\hline

\begin{minipage}{0.475\textwidth}\centering \resizebox{\textwidth}{!}{
$\begin{array}{l@{\hspace{0.1cm}}r@{\hspace{0.1cm}}ll}
{\rm set}& \multicolumn{3}{l}{S_3=\{\{1,2,3\},\{1,3,4\},\{2,3,4\}\}} \\
\multirow{4}{*}{parents}&\ldelim \{ {4}{3 mm} 
 &\{\{1,3\},\{2,3,4\}\}&Rule2\\	
& &\{\{1,2,3\},\{3,4\}\}&Rule2\\
& &\{\{1,2,3\},\{1,2,4\},\{1,3,4\},\{2,3,4\}\}&Rule1\\
& &\{\{1,3,4\},\{2,3\}\}&Rule2\\
\end{array}$}
\end{minipage}
&
\begin{minipage}{0.475\textwidth}\centering {\small
$\begin{array}{l@{\hspace{0.1cm}}r@{\hspace{0.1cm}}ll}
{\rm set}& \multicolumn{3}{l}{S_4 = \{\{1,3\},\{2,3\},\{3,4\}\}} \\
\multirow{3}{*}{children}&\ldelim \{ {3}{4 mm}
  &\{\{1,3\},\{2,3,4\}\}&Rule3\\
& &\{\{3,4\},\{1,2,3\}\}&Rule3\\
& &\{\{2,3\},\{1,3,4\}\}&Rule3\\
		\end{array}$}
\end{minipage}\\
\begin{minipage}{0.45\textwidth}
		\resizebox{\textwidth}{!}{ 
			\begin{tikzpicture}
			\tikzstyle{term} = [draw,shape=rectangle,minimum height=2em,text centered, fill, color=gray!20]
			\tikzstyle{term2} = [draw,shape=rectangle,minimum height=2em,text centered, fill, color=orange]
			\tikzstyle{term3} = [draw,shape=rectangle,minimum height=2em,text centered, fill, color=green!50]
			\tikzstyle{term32} = [draw,ellipse,minimum height=2em,text centered, fill, color=green!50]
			\tikzstyle{term4} = [draw,shape=rectangle,minimum height=2em,text centered, fill, color=violet!50]
			\tikzstyle{term42} = [draw,ellipse,minimum height=2em,text centered, fill, color=violet!50]
			\tikzstyle{term5} = [draw,shape=rectangle,minimum height=2em,text centered, fill, color=blue!50]
				\tikzstyle{link} = [-]
				\def\xspace{1.5}
				\def\yspace{-1.5}
				
				\node[term] (1234) at (0,0){\textcolor{black}{$\{1,2,3,4\}$}};
				\node[term2](123) at (-2.5*\xspace,\yspace) {\textcolor{black}{$\{1,2,3\}$}};
				\node[term2](124) at (-1*\xspace,\yspace) {\textcolor{black}{$\{1,2,4\}$}};
				\node[term](134) at (1*\xspace,\yspace) {\textcolor{black}{$\{1,3,4\}$}};
				\node[term](234) at (2.5*\xspace,\yspace) {\textcolor{black}{$\{2,3,4\}$}};
				\node[term42](12) at (-4*\xspace,2*\yspace) {\textcolor{black}{$\{1,2\}$}};	
				\node[term42](13) at (-2.5*\xspace,2*\yspace) {\textcolor{black}{$\{1,3\}$}};
				\node[term42](14) at (-1*\xspace,2*\yspace) {\textcolor{black}{$\{1,4\}$}};
				\node[term42](23) at (1*\xspace,2*\yspace) {\textcolor{black}{$\{2,3\}$}};
				\node[term42](24) at (2.5*\xspace,2*\yspace) {\textcolor{black}{$\{2,4\}$}};	
				\node[term2](34) at (4*\xspace,2*\yspace) {\textcolor{black}{$\{3,4\}$}};	
				\node[term4](1) at (-2.5*\xspace,3*\yspace) {\textcolor{black}{$\{1\}$}};
				\node[term4](2) at (-1*\xspace,3*\yspace) {\textcolor{black}{$\{2\}$}};
				\node[term4](3) at (1*\xspace,3*\yspace) {\textcolor{black}{$\{3\}$}};
				\node[term4](4) at (2.5*\xspace,3*\yspace) {\textcolor{black}{$\{4\}$}};
				\node[term4](0) at (0,4*\yspace) {\textcolor{black}{$\{\emptyset\}$}};
				
				\draw[link] (1234) edge (123);
				\draw[link] (1234) edge (124);
				\draw[link] (1234) edge (134);	
				\draw[link] (1234) edge (234);	
				\draw[link] (123) edge (12);
				\draw[link] (123) edge (13);
				\draw[link] (123) edge (23);
				\draw[link] (124) edge (12);
				\draw[link] (124) edge (14);
				\draw[link] (124) edge (24);	
				\draw[link] (134) edge (13);
				\draw[link] (134) edge (14);
				\draw[link] (134) edge (34);
				\draw[link] (234) edge (23);
				\draw[link] (234) edge (34);
				\draw[link] (234) edge (24);	
				\draw[link] (1) edge (12);
				\draw[link] (1) edge (13);
				\draw[link] (1) edge (14);
				\draw[link] (2) edge (12);
				\draw[link] (2) edge (23);
				\draw[link] (2) edge (24);
				\draw[link] (3) edge (13);
				\draw[link] (3) edge (23);
				\draw[link] (3) edge (34);
				\draw[link] (4) edge (14);
				\draw[link] (4) edge (24);
				\draw[link] (4) edge (34);
				\draw[link] (0) edge (1);
				\draw[link] (0) edge (2);
				\draw[link] (0) edge (3);
				\draw[link] (0) edge (4);
		\end{tikzpicture}}
	\end{minipage}&\begin{minipage}{0.45\textwidth}
		\resizebox{\textwidth}{!}{ 
			\begin{tikzpicture}
			\tikzstyle{term} = [draw,shape=rectangle,minimum height=2em,text centered, fill, color=gray!20]
			\tikzstyle{term2} = [draw,shape=rectangle,minimum height=2em,text centered, fill, color=orange]
			\tikzstyle{term3} = [draw,shape=rectangle,minimum height=2em,text centered, fill, color=green!50]
			\tikzstyle{term32} = [draw,ellipse,minimum height=2em,text centered, fill, color=green!50]
			\tikzstyle{term4} = [draw,shape=rectangle,minimum height=2em,text centered, fill, color=violet!50]
			\tikzstyle{term42} = [draw,ellipse,minimum height=2em,text centered, fill, color=violet!50]
			\tikzstyle{term5} = [draw,shape=rectangle,minimum height=2em,text centered, fill, color=blue!50]
				\tikzstyle{link} = [-]
				\def\xspace{1.5}
				\def\yspace{-1.5}
				
				\node[term] (1234) at (0,0){\textcolor{black}{$\{1,2,3,4\}$}};
				\node[term](123) at (-2.5*\xspace,\yspace) {\textcolor{black}{$\{1,2,3\}$}};
				\node[term32](124) at (-1*\xspace,\yspace) {\textcolor{black}{$\{1,2,4\}$}};
				\node[term](134) at (1*\xspace,\yspace) {\textcolor{black}{$\{1,3,4\}$}};
				\node[term](234) at (2.5*\xspace,\yspace) {\textcolor{black}{$\{2,3,4\}$}};
				\node[term3](12) at (-4*\xspace,2*\yspace) {\textcolor{black}{$\{1,2\}$}};	
				\node[term2](13) at (-2.5*\xspace,2*\yspace) {\textcolor{black}{$\{1,3\}$}};
				\node[term3](14) at (-1*\xspace,2*\yspace) {\textcolor{black}{$\{1,4\}$}};
				\node[term2](23) at (1*\xspace,2*\yspace) {\textcolor{black}{$\{2,3\}$}};
				\node[term3](24) at (2.5*\xspace,2*\yspace) {\textcolor{black}{$\{2,4\}$}};	
				\node[term2](34) at (4*\xspace,2*\yspace) {\textcolor{black}{$\{3,4\}$}};	
				\node[term42](1) at (-2.5*\xspace,3*\yspace) {\textcolor{black}{$\{1\}$}};
				\node[term42](2) at (-1*\xspace,3*\yspace) {\textcolor{black}{$\{2\}$}};
				\node[term42](3) at (1*\xspace,3*\yspace) {\textcolor{black}{$\{3\}$}};
				\node[term42](4) at (2.5*\xspace,3*\yspace) {\textcolor{black}{$\{4\}$}};
				\node[term4](0) at (0,4*\yspace) {\textcolor{black}{$\{\emptyset\}$}};
				
				\draw[link] (1234) edge (123);
				\draw[link] (1234) edge (124);
				\draw[link] (1234) edge (134);	
				\draw[link] (1234) edge (234);	
				\draw[link] (123) edge (12);
				\draw[link] (123) edge (13);
				\draw[link] (123) edge (23);
				\draw[link] (124) edge (12);
				\draw[link] (124) edge (14);
				\draw[link] (124) edge (24);	
				\draw[link] (134) edge (13);
				\draw[link] (134) edge (14);
				\draw[link] (134) edge (34);
				\draw[link] (234) edge (23);
				\draw[link] (234) edge (34);
				\draw[link] (234) edge (24);	
				\draw[link] (1) edge (12);
				\draw[link] (1) edge (13);
				\draw[link] (1) edge (14);
				\draw[link] (2) edge (12);
				\draw[link] (2) edge (23);
				\draw[link] (2) edge (24);
				\draw[link] (3) edge (13);
				\draw[link] (3) edge (23);
				\draw[link] (3) edge (34);
				\draw[link] (4) edge (14);
				\draw[link] (4) edge (24);
				\draw[link] (4) edge (34);
				\draw[link] (0) edge (1);
				\draw[link] (0) edge (2);
				\draw[link] (0) edge (3);
				\draw[link] (0) edge (4);
		\end{tikzpicture}}
	\end{minipage}\\\hline

\multicolumn{2}{|c|}{\begin{minipage}{0.95\textwidth}
\resizebox{\textwidth}{!}{ 
\begin{tikzpicture}
	\tikzstyle{forange} = [draw,color=orange,fill]
	\tikzstyle{fwhite} = [draw]
	\def\xspace{2.8}
	\def\yspace{-2.7}

\node[forange] (0) at (0,4*\yspace) 
    {\textcolor{black}{$\{1,2,3\},\{1,3,4\},\{2,3,4\}_{S1}$}};
\node[fwhite] (1) at (-2.5*\xspace,3*\yspace) 
    {\textcolor{black}{$\{1,3\},\{2,3,4\}$}};
\node[forange] (2) at (-1*\xspace,3*\yspace) 
    {\textcolor{black}{$\{1,2,3\},\{3,4\}_{S2}$}};
\node[fwhite]  (3) at (0.75*\xspace,3*\yspace) 
    {\textcolor{black}{$\{1,2,3\},\{1,2,4\},\{1,3,4\},\{2,3,4\}$}};
\node[fwhite] (4) at (2.5*\xspace,3*\yspace) 
    {\textcolor{black}{$\{1,3,4\},\{2,3\}$}};
\node[fwhite]  (5) at (-4*\xspace,2*\yspace) 
    {\textcolor{black}{$\{1,2,4\},\{1,3\},\{2,3,4\}$}};
\node[fwhite]  (6) at (-2.5*\xspace,2*\yspace) 
    {\textcolor{black}{$\{1,2\},\{1,3,4\},\{2,3,4\}$}};
\node[forange]  (7) at (-1*\xspace,2*\yspace) 
    {\textcolor{black}{$\{1,2,3\},\{1,2,4\},\{3,4\}_{S3}$}};
\node[fwhite]  (8) at (1*\xspace,2*\yspace) 
    {\textcolor{black}{$\{1,2,3\},\{1,4\},\{2,3,4\}$}};
\node[fwhite]  (9) at (2.5*\xspace,2*\yspace) 
    {\textcolor{black}{$\{1,2,4\},\{1,3,4\},\{2,3\}$}};
\node[fwhite] (10) at (4*\xspace,2*\yspace) 
    {\textcolor{black}{$\{1,2,3\},\{1,3,4\},\{2,4\}$}};
\node[fwhite]  (11) at (-1.5*\xspace,\yspace)
    {\textcolor{black}{$\{1,2,4\},\{1,3\},\{2,3\}$}};
\node[forange] (12) at (0*\xspace,\yspace) 
    {\textcolor{black}{$\{1,3\},\{2,3\},\{3,4\}_{S4}$}};
\node[fwhite]  (14) at (1.5*\xspace,\yspace)
    {\textcolor{black}{$\{1,2,4\},\{1,3\},\{3,4\}$}};
\node[fwhite]  (15) at (3*\xspace,\yspace)
    {\textcolor{black}{$\{1,2,4\},\{2,3\},\{3,4\}$}};
\node[fwhite] (13) at (0*\xspace,0*\yspace)
    {\textcolor{black}{$\{1,2,4\},\{1,3\},\{2,3\},\{3,4\}$}};

\draw[-] (0) edge node[sloped,above] {\bf R2} (1);
\draw[-] (0) edge node[sloped,above] {\bf R2} (2);
\draw[-] (0) edge node[sloped,above] {\bf R1} (3);
\draw[-] (0) edge node[sloped,above] {\bf R2} (4);
\draw[-] (1) edge node[pos=0.3,sloped,above] {\bf R1} (5);
\draw[-] (2) edge node[pos=0.8] {\bf R1} (7);
\draw[-] (4) edge node[pos=0.2] {\bf R1} (9);
\draw[-] (3) edge node[pos=0.2] {\bf R2} (5);
\draw[-] (3) edge node[pos=0.2] {\bf R2} (6);
\draw[-] (3) edge node[pos=0.2] {\bf R2} (7);
\draw[-] (3) edge node[pos=0.2] {\bf R2} (8);
\draw[-] (3) edge node[pos=0.2] {\bf R2} (9);
\draw[-] (3) edge node[pos=0.2] {\bf R2} (10);
\draw[-] (1) edge node[pos=0.1,sloped,above] {\bf R3} (12);
\draw[-] (2) edge node[pos=0.75,sloped,above] {\bf R3} (12);
\draw[-] (4) edge node[pos=0.85,sloped,above] {\bf R3} (12);
\draw[-] (5) edge node[sloped,above] {\bf R2} (11);
\draw[-] (5) edge node[sloped,above] {\bf R2} (14);
\draw[-] (7) edge node[sloped,above] {\bf R2} (15);
\draw[-] (9) edge node[pos=0.1,sloped,above] {\bf R2} (11);
\draw[-] (9) edge node[sloped,above] {\bf R2} (15);
\draw[-] (11) edge node[above] {\bf R1} (13);
\draw[-] (12) edge node[left]  {\bf R1} (13);
\draw[-] (14) edge node[above] {\bf R1} (13);
\draw[-] (15) edge node[above] {\bf R1} (13);


\draw[dotted,thick] (0) -- (-0.2*\xspace,4.3*\yspace);
\draw[dotted,thick] (0) -- (0*\xspace,4.3*\yspace);
\draw[dotted,thick] (0) -- (0.2*\xspace,4.3*\yspace);

\draw[dotted,thick] (1) -- (-2.7*\xspace,3.3*\yspace);

\draw[dotted,thick] (4) -- (2.3*\xspace,3.3*\yspace);
\draw[dotted,thick] (4) -- (2.5*\xspace,3.3*\yspace);
\draw[dotted,thick] (4) -- (2.7*\xspace,3.3*\yspace);

\draw[dotted,thick] (5) -- (-4.2*\xspace,2.3*\yspace);
\draw[dotted,thick] (5) -- (-4*\xspace,2.3*\yspace);
\draw[dotted,thick] (5) -- (-3.8*\xspace,2.3*\yspace);

\draw[dotted,thick] (10) -- (3.8*\xspace,2.3*\yspace);
\draw[dotted,thick] (10) -- (4*\xspace,2.3*\yspace);
\draw[dotted,thick] (10) -- (4.2*\xspace,2.3*\yspace);

\draw[dotted,thick] (11) -- (-1.6*\xspace,1.3*\yspace);

\draw[dotted,thick] (11) -- (-1.8*\xspace,0.7*\yspace);
\draw[dotted,thick] (11) -- (-1.6*\xspace,0.7*\yspace);
\draw[dotted,thick] (11) -- (-1.4*\xspace,0.7*\yspace);

\draw[dotted,thick] (5) -- (-4.3*\xspace,1.7*\yspace);
\draw[dotted,thick] (5) -- (-4.1*\xspace,1.7*\yspace);
\draw[dotted,thick] (5) -- (-3.9*\xspace,1.7*\yspace);
\draw[dotted,thick] (5) -- (-3.7*\xspace,1.7*\yspace);

\draw[dotted,thick] (6) -- (-2.9*\xspace,1.7*\yspace);
\draw[dotted,thick] (6) -- (-2.7*\xspace,1.7*\yspace);
\draw[dotted,thick] (6) -- (-2.5*\xspace,1.7*\yspace);
\draw[dotted,thick] (6) -- (-2.3*\xspace,1.7*\yspace);
\draw[dotted,thick] (6) -- (-2.1*\xspace,1.7*\yspace);

\draw[dotted,thick] (7) -- (-1.3*\xspace,1.7*\yspace);
\draw[dotted,thick] (7) -- (-1.1*\xspace,1.7*\yspace);
\draw[dotted,thick] (7) -- (-0.9*\xspace,1.7*\yspace);
\draw[dotted,thick] (7) -- (-0.7*\xspace,1.7*\yspace);

\draw[dotted,thick] (7) -- (-1.3*\xspace,2.3*\yspace);

\draw[dotted,thick] (8) -- (0.6*\xspace,1.7*\yspace);
\draw[dotted,thick] (8) -- (0.8*\xspace,1.7*\yspace);
\draw[dotted,thick] (8) -- (1.0*\xspace,1.7*\yspace);
\draw[dotted,thick] (8) -- (1.2*\xspace,1.7*\yspace);
\draw[dotted,thick] (8) -- (1.4*\xspace,1.7*\yspace);

\draw[dotted,thick] (9) -- (2.3*\xspace,1.7*\yspace);
\draw[dotted,thick] (9) -- (2.5*\xspace,1.7*\yspace);
\draw[dotted,thick] (9) -- (2.8*\xspace,1.7*\yspace);

\draw[dotted,thick] (10) -- (3.6*\xspace,1.7*\yspace);
\draw[dotted,thick] (10) -- (3.8*\xspace,1.7*\yspace);
\draw[dotted,thick] (10) -- (4*\xspace,1.7*\yspace);
\draw[dotted,thick] (10) -- (4.2*\xspace,1.7*\yspace);
\draw[dotted,thick] (10) -- (4.4*\xspace,1.7*\yspace);

\draw[dotted,thick] (13) -- (-0.2*\xspace,0.3*\yspace);
\draw[dotted,thick] (13) -- (0*\xspace,0.3*\yspace);
\draw[dotted,thick] (13) -- (0.2*\xspace,0.3*\yspace);

\draw[dotted,thick] (14) -- (1.4*\xspace,0.7*\yspace);
\draw[dotted,thick] (14) -- (1.6*\xspace,0.7*\yspace);

\draw[dotted,thick] (14) -- (1.4*\xspace,1.3*\yspace);
\draw[dotted,thick] (14) -- (1.6*\xspace,1.3*\yspace);

\draw[dotted,thick] (13) -- (-0.3*\xspace,-0.3*\yspace);
\draw[dotted,thick] (13) -- (-0.1*\xspace,-0.3*\yspace);
\draw[dotted,thick] (13) -- (0.1*\xspace,-0.3*\yspace);
\draw[dotted,thick] (13) -- (0.3*\xspace,-0.3*\yspace);

\draw[dotted,thick] (15) -- (3.2*\xspace,0.7*\yspace);
\draw[dotted,thick] (15) -- (3.0*\xspace,0.7*\yspace);
\draw[dotted,thick] (15) -- (2.8*\xspace,0.7*\yspace);

\draw[dotted,thick] (15) -- (3.2*\xspace,1.3*\yspace);

\end{tikzpicture}}
\end{minipage}} \\ \hline
\end{tabular}}
\end{center}
\caption{\footnotesize{Complete illustration of the rules to compute parents (left) and children (right). Functions using the set representation on the Hasse diagram of the poset $(2^{\{1,2,3,4\}},\subseteq)$.
Orange vertices depict the elements of the starting sets ($S_i, i=1,2,3,4$); green vertices depict  independent sets, among which oval ones are maximal; violet vertices depict dominated sets, among which oval ones are maximal and not included in any independent set.
Subset of the Hasse diagram (below) representing the partial order between the functions presented, as well as some of their neighbouring relationships.}
   \label{fig:HD-rules}}
\end{figure}

\newpage
\section{\label{sec:app:algorithms}Algorithms}

Here, we provide a detailed description of Algorithm~\ref{alg:parents} and Algorithm~\ref{alg:children}.
These algorithms are implemented as a dedicated Python library and made available \url{https://github.com/ptgm/pyfunctionhood}, under the GNU General Public License v3.0 (GPL-3.0).

\subsection{Algorithm to compute immediate parents}\label{sec:app:alg_parents}

To compute the immediate parents, Algorithm~\ref{alg:parents} starts by computing the set $\mathcal{C}$ of maximal sets independent of $S$ (line \ref{line:p:indep}).
Next, lines~\ref{line:p:for_indep}--\ref{line:p:for_indep_end} iterate over the maximal set $\mathcal{C}$, generating one new immediate parent for each $c \in \mathcal{C}$, following Rule 1. Each new immediate parent generated is kept in $\mathcal{P}$.
Next, we compute the set $\mathcal{D}$ of maximal sets dominated by $S$ (line~\ref{line:p:maxdom}), where each element of $\mathcal{D}$ is not contained in another element of $\mathcal{D}$. Lines~\ref{line:p:maxdom_iter}--\ref{line:p:maxdom_iter_end} ensure that each element of $\mathcal{D}$ is not contained in any element of $\mathcal{C}$, {\it i.e.}, in any maximal set independent of $S$. Elements of $\mathcal{D}$ are candidates to be considered afterwards in the generation of immediate parents by Rule 2 or by Rule 3. 

Line~\ref{line:p:map} initialises a map relating elements $s \in S$ (the map keys) with elements of $d \in \mathcal{D}$ (the map values) which are included in $s$ and were not used by Rule 2.
Lines~\ref{line:p:iterD}--\ref{line:p:iterD_end} iterate over each element $d \in \mathcal{D}$, trying to find one complying with conditions from Rule 2. If $S$ deprived with elements containing $d$ and augmented with $d$ remains a cover, it is generated as a new immediate parent by Rule 2. Otherwise, the element $d$ is added to the map $\mathcal{D}notused$, associated with every set of $S$ that contains it. These elements in $\mathcal{D}notused$ will be candidates for Rule 3. 

Finally, lines~\ref{line:p:iterNotUsed}--\ref{line:p:iterNotUsed_end} iterate over the elements of $S$, the keys of $\mathcal{D}notused$, with associated elements of $\mathcal{D}$ that were not used by Rule 2, due to insufficient cover.
The two inner for loops (lines~\ref{line:p:innerLoop}--\ref{line:p:innerLoop_end}) do a pairwise combination of these associated elements, ensuring cover and therefore generating a new immediate parent by Rule 3.
If a given $s \in S$ is not present as a key in the map $\mathcal{D}notused$, or if it is present with a single associated $d$, then it has no immediate parent generated by Rule 3.

\subsection{Algorithm to compute immediate children}\label{sec:app:alg_children}

To compute the immediate children, Algorithm~\ref{alg:children} iterates over each element $s \in S$ (lines~\ref{line:c:iterS}--\ref{line:c:iterS_end}), considering $S \setminus \{s\}$ as the working child candidate.
It then iterates over each of the missing literals in $s$ (lines~\ref{line:c:missingLits}--\ref{line:c:missingLits_end}) to verify if each new element $s \cup \{l\}$ contains one, two or more elements of $S$.

If $s \cup \{l\}$ contains only one element, it has to be $s$ itself, meaning that the working child candidate can be extended with the new element $s \cup \{l\}$. The iteration over all missing literals ensures that each time the child candidate is extended with $s \cup \{l\}$. In this case, the flag $isExtendable$ is set to add the extended child candidate as a valid immediate child (lines~\ref{line:c:extendable}--\ref{line:c:extendable_end}).

On the other hand, if $s \cup \{l\}$ contains exactly two elements of $S$, $s$ and another element in $S$, it means that these two elements of $S$ can potentially be replaced by $s \cup \{l\}$. The $mergeable$ set keeps all these elements grouped by size (line~\ref{line:c:mergeable}), as candidates to be used by Rule 3.

Finally, if $s \cup \{l\}$ contains more than two elements of $S$, does not fall into any of the previous cases for every missing literal, and is still a cover, then it is because $s$ is a maximal set independent of $S \setminus \{s\}$.
The working child candidate $S \setminus \{s\}$ is then considered a valid immediate child (line~\ref{line:c:rule2}).

Lines~\ref{line:c:mergeableIter}--\ref{line:c:mergeableIter_end} iterate over the sizes of $mergeable$ elements of $S$ found previously. For each existing size $sz$, if extending one of the elements in $mergeable[sz]$ with a given literal $\{l\}$ contains exactly two elements, Rule 3 can be applied, considering $S \setminus absorbed \cup \{ s \cup \{l\}\}$ as a new valid immediate child (line~\ref{line:c:rule3}).

{\footnotesize \begin{algorithm}[t]
    \caption{Algorithm to Compute Immediate Children}
    \label{alg:children}
    \footnotesize
    \textbf{Input:} An element $S \in \mathcal{S}_p$\\
    \textbf{Output:} Set $\mathcal{C} \subseteq \mathcal{S}_p$ with all immediate children of $S$\\
    \begin{algorithmic}[1] 
        \STATE {\em getContainedBy($s$,$S$)} : set of all sets in $S$ contained in $s$
        \STATE {\em getMissingLits($s$)} : set of literals in $\{1,\dots, p\}$ missing in $s$
        \STATE {\em isCover($S$)} : True if $S$ is a cover of $\{1,\dots, p\}$

        \STATE $\mathcal{C} \longleftarrow \varnothing $
        \FOR{$s \in S$} \label{line:c:iterS}
            \STATE $toMerge \longleftarrow$ \texttt{false}
            \STATE $isExtendable \longleftarrow$ \texttt{false}
            \STATE $childCandidate \longleftarrow S \setminus \{s\}$
            \FOR{$l \in getMissingLits(s)$} \label{line:c:missingLits}
                \STATE $contained \longleftarrow getContainedBy(s \cup \{l\},S)$
                \IF{$|contained|=1$}
                    \STATE $isExtendable \longleftarrow$ \texttt{true} \label{line:c:extendable}
                    \STATE $childCandidate \longleftarrow childCandidate \cup \{s \cup \{l\}\}$ \label{line:c:extendable_end}
                \ENDIF
                \IF{$|absorbed|=2$}
                    \STATE $toMerge \longleftarrow$ \texttt{true}
                \ENDIF
            \ENDFOR \label{line:c:missingLits_end}
            \IF{$isExtendable$}
                \STATE $\mathcal{C} \longleftarrow \mathcal{C} \cup \{childCandidate\}$ \label{line:c:rule2} \hfill // Rule 2
            \ELSIF{$toMerge$}
                \STATE $mergeable[len(s)] \longleftarrow mergeable[len(s)] \cup \{s\}$ \label{line:c:mergeable} 
            \ENDIF
        \ENDFOR \label{line:c:iterS_end}
        \FOR{$sz \in mergeable$} \label{line:c:mergeableIter}
            \WHILE{$mergeable[sz] \neq \emptyset$}
                \STATE $s \longleftarrow mergeable[sz].top()$
                \FOR{$l \in getMissingLits(s)$}
                    \STATE $cnt \longleftarrow getContainedBy(s \cup \{l\},mergeable[sz])$
                    \IF{$|cnt|=2$}
                        \STATE $\mathcal{C}\longleftarrow \mathcal{C} \cup \{S\setminus cnt \cup \{s \cup \{l\}\}\}$ \label{line:c:rule3} \hfill // Rule 3
                    \ENDIF
                \ENDFOR
                \STATE $mergeable[sz].pop()$
            \ENDWHILE
        \ENDFOR \label{line:c:mergeableIter_end}
        \RETURN $\mathcal{C}$
    \end{algorithmic}
\end{algorithm}
}

\end{document}